\documentclass[12pt,a4paper]{article}
\usepackage[top=1 in, bottom=1 in, left=1 in, right=1 in, letterpaper]{geometry}
\usepackage{amssymb}
\usepackage{amsmath}
\usepackage{amsthm}
\usepackage{xspace}
\usepackage{algorithm}
\usepackage[noend]{algorithmic}
\usepackage{fancyhdr}
\usepackage{graphicx}
\usepackage{enumerate}
\usepackage{multirow}

\newtheorem{theorem}{Theorem}

\newtheorem{proposition}{Proposition}
\newtheorem{lemma}{Lemma}

\title{\textbf{Energy Efficient Scheduling and Routing via Randomized Rounding}}

\author{\\
Evripidis Bampis\footnote{Sorbonne Universit\'{e}s, UPMC Univ Paris 06, UMR 7606, LIP6, F-75005, France \newline\emph{\{Evripidis.Bampis, Giorgio.Lucarelli, Dimitrios.Letsios\}@lip6.fr} \newline}\hspace*{1cm}
Alexander Kononov\footnote{Sobolev Institute of Mathematics, Novosibirsk, Russia \newline
\emph{alvenko@math.nsc.ru}\newline} \hspace*{1cm}
Dimitrios Letsios$^*$ \\ \\ Giorgio Lucarelli$^*$ \hspace*{1cm}
Maxim Sviridenko\footnote{Yahoo\! Labs, New York, NY \newline \emph{sviri@yahoo-inc.com}}
}

\date{}

\begin{document}

\maketitle

\begin{abstract}
We propose a unifying framework based on configuration linear programs and randomized rounding,
for different energy optimization problems in the dynamic speed-scaling setting.
We apply our framework to various scheduling and routing problems in heterogeneous computing and networking environments.
We first consider the energy minimization problem of scheduling a set of jobs on a set of parallel speed scalable processors in a fully heterogeneous setting.
For both the preemptive-non-migratory and the preemptive-migratory variants,
our approach allows us to obtain solutions of almost the same quality as for the homogeneous environment.
By exploiting the  result for the preemptive-non-migratory variant, we are able
to improve the best known approximation ratio for the single processor non-preemptive problem.
Furthermore, we show that our approach allows to obtain a constant-factor approximation algorithm for the power-aware preemptive job shop scheduling problem.
Finally, we consider the min-power routing problem
where we are given a network modeled by an undirected graph and a set of uniform demands
that have to be routed on integral routes from their sources to their destinations so that the energy consumption is minimized.
We improve the best known approximation ratio for this problem.
\end{abstract}

\section{Introduction}

We focus on energy minimization problems in heterogeneous computing and networking environments in the dynamic speed-scaling setting.
For many years, the exponential increase of processors' frequencies followed Moore's law.
This is no more possible because of physical (thermal) constraints.
Today, for improving the performance of modern computing systems, designers use parallelism,
i.e., multiple cores running at lower frequencies but offering better performances than a single core.
These systems can be either homogeneous where an identical core is used many times, or heterogeneous combining general-purpose and special-purpose cores.
Heterogeneity offers the possibility of further improving the performance of the system
by executing each job on the most appropriate type of processors \cite{BrodtkorbDHHS10}.
However in order to exploit the opportunities offered by the heterogeneous systems,
it is essential to focus on the design of new efficient power-aware algorithms taking into account the heterogeneity of these architectures.
In this direction, Gupta et al. in \cite{GuptaIKMP12} have studied the impact of the introduction of the heterogeneity
on the difficulty of various power-aware scheduling problems.

In  this paper, we show that rounding configuration linear programs  helps in handling the heterogeneity of both the jobs and the processors.
We adopt one of the main mechanisms for reducing the energy consumption in modern computer systems which is based on the use of speed scalable processors.
Starting from the seminal paper of Yao et al. \cite{YaoDS95}, many papers adopted the speed-scaling model
in which if a processor runs at speed $s$, then the rate of the energy consumption, i.e., the power, is $P(s)=s^{\alpha}$ with $\alpha$ a constant close to 3
(new studies show that $\alpha$ is rather smaller: 1.11 for Intel PXA 270, 1.62 for Pentium M770 and 1.66 for a TCP offload engine \cite{WiermanAT09}).
Moreover, the energy consumption is the integral of the power over time.
This model captures the intuitive idea that the faster a processor works the more energy it consumes.

We first consider a {\em fully heterogeneous environment} where both
the jobs' characteristics are processor-dependent and every processor has its own power function.
Formally, we consider the following problem: we are given a set $\mathcal{J}$ of $n$ jobs and a set $\mathcal{P}$ of $m$ parallel processors.
Every processor $i \in \mathcal{P}$ obeys to a different speed-to-power function, i.e., it is associated with a different $\alpha_i \geq 1$
and hence if a job runs at speed $s$ on processor $i$, then the power is $P(s)=s^{\alpha_i}$.
Each job $j \in \mathcal{J}$ has a different release date $r_{i,j}$, deadline $d_{i,j}$ and workload $w_{i,j}$
if it is executed on processor $i \in \mathcal{P}$.
The goal is to find a schedule of minimum energy respecting the release dates and the deadlines of the jobs.

The assumption that the jobs have processor-dependent works
covers for example the problem of scheduling in the restricted assignment model (see \cite{Svensson11}).
In this model each job is associated with a subset of processors and has to be executed on one of them.
Clearly, in our model the work of each job is the same on the processors of its corresponding subset and infinite on the remaining processors.
Moreover, processor-dependent release dates have been already studied in the literature when the processors are connected by a network.
In such a case, it is assumed that every job is initially available at a given processor
and that a transfer time must elapse before it becomes available at a new machine \cite{AwerbuchKP92,DengLX90}.

In this paper we propose a unifying framework for minimizing energy in different heterogeneous computing and networking environments.
We first consider two variants of the heterogeneous multiprocessor {\em preemptive} problem.
In both cases, the execution of a job may be interrupted and resumed later.
In the {\em non-migratory} case each job has to be entirely executed on a single processor.
In the {\em migratory} case each job may be executed by more than one processors, without allowing parallel execution of a job.
We also focus on the {\em non-preemptive} single processor case.
Furthermore, we consider the energy minimization problem in an heterogeneous \emph{job shop} environment where the jobs can be preempted.
Finally, we study the {\em min-power routing problem}, introduced in \cite{AndrewsAZZ12},
where a set of uniform demands have to be routed on integral routes from their sources to
their destinations so that the energy consumption to be minimized.
We believe that our general techniques will find further applications in energy optimization.

\subsection{Related Work}
Yao et al. \cite{YaoDS95} proposed an optimal algorithm for finding a feasible preemptive schedule
with minimum energy consumption when a single processor is available.

The homogeneous multiprocessor case has been solved optimally in polynomial time when both the preemption and the migration of jobs are allowed
\cite{AlbersAG11,AngelBKL12,BampisLL12,BinghamG08}.
Albers et al. \cite{AlbersMS07} considered the homogeneous multiprocessor preemptive problem, where the migration of the jobs is not allowed.
They proved that the problem is $\mathcal{NP}$-hard even for instances with common release dates and common deadlines.
Greiner et al. \cite{GreinerNS09} gave a generic reduction transforming an optimal schedule for the homogeneous multiprocessor problem with migration,
to a $B_{\lceil\alpha\rceil}$-approximate solution for the homogeneous multiprocessor preemptive problem without migration,
where $B_{\lceil\alpha\rceil}$ is the $\lceil\alpha\rceil$-th Bell number.

Antoniadis and Huang \cite{AntoniadisH12} proved that the single processor non-preemptive problem is $\mathcal{NP}$-hard
even for instances in which for any two jobs $j$ and $j'$ with $r_j \leq r_{j'}$ it holds that $d_j \geq d_{j'}$.
They also proposed a $2^{5\alpha-4}$-approximation algorithm for general instances.
For the homogeneous multiprocessor non-preemptive case an approximation algorithm of ratio
$m^{\alpha}(\sqrt[m]{n})^{\alpha-1}$ has been proposed in \cite{BampisKLLN13}.

Andrews et al. \cite{AndrewsAZZ12} studied the min-power routing problem and for uniform demands,
i.e. for the case where all the demands have the same value, they proposed a $\gamma$-approximation algorithm,
where $\gamma=\max\{1+\tau 2^{\alpha (\tau+1) \log e}, 2+ \tau 2^{\alpha(\tau+1)}\}$, with $\tau=\lceil 2 \log (\alpha+4)\rceil$.
For non-uniform demands, they proposed a $O(\log^{\alpha-1} D)$-approximation algorithm, where $D$ is the maximum value of the demands.

For further results on energy-efficient scheduling we refer the interested reader to the reviews \cite{Albers10,Albers11}.

\subsection{Notation}
We denote by $E(\mathcal{S})$ the total energy consumed by a schedule $\mathcal{S}$.
Moreover, we denote by $\mathcal{S}^*$ an optimal schedule and by $OPT$ the energy consumption of $\mathcal{S}^*$.
For each job $j \in \mathcal{J}$, we say that $j$ is \emph{alive} on processor $i \in \mathcal{P}$ during the interval $[r_{i,j},d_{i,j}]$.
Let $\alpha=\max_{i \in \mathcal{P}}\{\alpha_i\}$.
The Bell number, $B_n$, is defined for any integer $n\ge 0$ and corresponds to the number of partitions of a set of $n$ items.
It is well known that Bell numbers satisfy the following equality
\begin{equation*}
B_{n}=\sum_{k=0}^{\infty}\frac{k^{n} e^{-1}}{k!}
\end{equation*}
known as  Dobinski's formula.
Another way to state this formula is that the $n$-th Bell number is equal to the $n$-th moment of a Poisson random variable with parameter (expected value) $1$.
This naturally leads to a more general definition.
The generalized Bell number, denoted by $\tilde{B}_{\alpha}=\sum_{k=0}^{\infty}\frac{k^{\alpha} e^{-1}}{k!}$,
is defined for any $\alpha \in \mathbb{R}^+$ and corresponds to the $\alpha$-th (fractional) moment of a Poisson random variable with parameter $1$.
Note that the ratios of our algorithms depend on the generalized Bell number,
while the previous known \cite{GreinerNS09} on the standard Bell number.

\subsection{Our Contribution}
In this paper we formulate heterogeneous scheduling and routing problems using \emph{configuration linear programs (LPs)}
and we apply \emph{randomized rounding}.
In Section~\ref{section:non-migration}, we consider the heterogeneous multiprocessor speed-scaling problem without migrations
and we propose an approximation algorithm of ratio $(1+\varepsilon)\tilde{B}_{\alpha}$.
As this LP has an exponential number of variables,
we give an alternative (compact) formulation of the problem using a polynomial number of variables
and we prove the equivalence between the two LP relaxations.
For real values of $\alpha$ our result improves the $B_{\lceil\alpha\rceil}$ approximation ratio of \cite{GreinerNS09} for the homogeneous case
to $(1+\varepsilon)\tilde{B}_{\alpha}$ for the fully heterogeneous environment that we consider here (see Table~\ref{tbl:results}).
In Section~\ref{section:migration}, using again a configuration LP formulation,
we present an algorithm for the heterogeneous multiprocessor speed-scaling problem with migration.
This algorithm returns a solution which is
within an additive factor of $\varepsilon$ far from the optimal solution
and runs in time polynomial to the size of the instance and to $1/\varepsilon$.
This result generalizes the results of \cite{AlbersAG11,AngelBKL12,BampisLL12,BinghamG08} from an homogeneous environment
to a fully heterogeneous environment.
In Section~\ref{section:single}, we transform the single processor speed-scaling problem without preemptions to
the heterogeneous multiprocessor problem without migrations and we give an approximation algorithm of ratio $2^{\alpha-1}(1+\varepsilon)\tilde{B}_{\alpha}$,
improving upon the previous known $2^{5\alpha-4}$-approximation algorithm in \cite{AntoniadisH12} for any $\alpha<114$ (see Table~\ref{tbl:results}).
In Section~\ref{section:jobshop}, we study the power-aware preemptive job shop scheduling problem and we propose a $((1+\varepsilon) \tilde{B}_{\alpha})$-approximation algorithm.
Finally, in Section~\ref{section:routing}, we improve the analysis for the min-power routing problem with uniform demands given in \cite{AndrewsAZZ12},
based on the randomized rounding analysis that we propose in this paper.
Our approach gives an approximation ratio of $\tilde{B}_{\alpha}$ significantly improving the analysis given in \cite{AndrewsAZZ12} (see Table~\ref{tbl:results}).

\begin{table}
\begin{center}
\begin{tabular}{c||c|c||c|c||c|c}
& \multicolumn{2}{|c||}{Multiprocessor} &\multicolumn{2}{|c||}{Non-Preemptive} & \multicolumn{2}{|c}{Routing}\\
Value of & \multicolumn{2}{|c||}{Non-Migratory} &\multicolumn{2}{|c||}{Single processor} & \multicolumn{2}{|c}{Uniform Demands}\\
\cline{2-7}
$\alpha$ & \cite{GreinerNS09} & Our & \multirow{2}{*}{\cite{AntoniadisH12}} & \multirow{2}{*}{Our} & \multirow{2}{*}{\cite{AndrewsAZZ12}} & \multirow{2}{*}{Our} \\
& Homogeneous & Heterogeneous & & & & \\
\hline
\hline
1.11 & 2 & 1.07(1+$\epsilon$) & 2.93  & 1.15(1+$\epsilon$) & 375    & 1.07 \\
1.62 & 2 & 1.49(1+$\epsilon$) & 17.15 & 2.30(1+$\epsilon$) & 2196   & 1.49 \\
1.66 & 2 & 1.54(1+$\epsilon$) & 19.70 & 2.43(1+$\epsilon$) & 2522   & 1.54 \\
2    & 2 & 2(1+$\epsilon$)    & 64    & 4(1+$\epsilon$)    & 8193   & 2    \\
2.5  & 5 & 3.08(1+$\epsilon$) & 362   & 8.72(1+$\epsilon$) & 46342  & 3.08 \\
3    & 5 & 5(1+$\epsilon$)    & 2048  & 20(1+$\epsilon$)   & 262145 & 5
\end{tabular}
\end{center}
\caption{Comparison of our approximation ratios vs. better previous known ratios for:
(i) the preemptive multiprocessor problem without migrations, (ii) the single processor non-preemptive problem, and (iii) the min-power routing problem.}
\label{tbl:results}
\end{table}

\section{Technical Probabilistic Propositions}

In this section, we state and prove a series of technical propositions which are key ingredients in our analysis.
Proposition \ref{Prop:ConcaveExpectation} bounds the expectation of a specific function of random variables.

\begin{proposition}
\label{Prop:ConcaveExpectation}
Consider $n$ random variables $X_1,X_2,\ldots, X_n$ and let $\alpha>1$. Then, it holds that
\begin{equation*}
\mathbb{E}\left[\left(\sum_{i=1}^nX_i^{1/\alpha}\right)^{\alpha}\right]\leq \left(\sum_{i=1}^n\mathbb{E}[X_i]^{1/\alpha}\right)^{\alpha}
\end{equation*}
\end{proposition}
\begin{proof}
By defining random variables $Y_i=X_i^{1/\alpha}$ and applying the Minkowski's inequality we derive
\begin{equation*}
\mathbb{E}\left[\left(\sum_{i=1}^nY_i\right)^{\alpha}\right]\leq \left(\sum_{i=1}^n\mathbb{E}[Y_i^{\alpha}]^{1/\alpha}\right)^{\alpha}=
\left(\sum_{i=1}^n\mathbb{E}[X_i]^{1/\alpha}\right)^{\alpha}
\end{equation*}
\end{proof}

Proposition~\ref{prop:split} deals with the expressions arising when one estimates the moments of random variables with Binomial distributions.

\begin{proposition}\label{prop:split}
Consider a set of non-negative constants $\{e_1,e_2,\ldots,e_n\}$.
For any subset $A$ of these constants, let $f(A)=\left(\sum_{j\in A}e_j^{1/\alpha}\right)^{\alpha}$.
Let $S$ be a random set generated by choosing each element $e_j$, $1\leq j\leq n$, independently at random with probability $Y_j$.
Moreover, let $e_{n}=e_{n+1}$.
Assume that $S'$ is a random set generated by sampling independently at random within $\{e_1,e_2,\ldots,e_{n+1}\}$,
with probabilities $Y_1',Y_2',\ldots,Y_{n+1}'$, where $Y_j'=Y_j$, for $1\leq j\leq n-1$, and $Y_n'+Y_{n+1}'=Y_n$.
Then,
\begin{equation*}
\mathbb{E}[f(S)]\leq\mathbb{E}[f(S')]
\end{equation*}
\end{proposition}
\begin{proof}
Let $Pr(T)$ be the probability that exactly the constants in the set $T$ are chosen among the constants in $U=\{e_1,e_2,\ldots,e_{n-1}\}$.
That is, $Pr(T)=\prod_{e_j\in T} Y_j\prod_{e_j\in U\setminus T}(1-Y_j)$.
We have that
\begin{eqnarray*}
\mathbb{E}[f(S)]
& = & \sum_{T\subseteq U}Pr(T)\Bigg[(1-Y_n)\cdot\left(\sum_{j\in T}e_j^{1/\alpha}\right)^{\alpha} + Y_n\cdot\left(\sum_{j\in T}e_j^{1/\alpha}+e_n^{1/\alpha}\right)^{\alpha}\Bigg]
\end{eqnarray*}
As $Y_j=Y_j'$ for $1\leq j\leq n-1$, it holds that
\begin{eqnarray*}
\mathbb{E}[f(S')]
& = & \sum_{T\subseteq U}Pr(T)\Bigg[(1-Y_n')\cdot(1-Y_{n+1}')\cdot\left(\sum_{j\in T}e_j^{1/\alpha}\right)^{\alpha}+ Y_n'\cdot(1-Y_{n+1}')\cdot\left(\sum_{j\in T}e_j^{1/\alpha}+e_n^{1/\alpha}\right)^{\alpha}\\
&  & + (1-Y_n')\cdot Y_{n+1}'\cdot\left(\sum_{j\in T}e_j^{1/\alpha}+e_{n+1}^{1/\alpha}\right)^{\alpha}+Y_n'\cdot Y_{n+1}'\cdot\left(\sum_{j\in T}e_j^{1/\alpha}+e_n^{1/\alpha} +e_{n+1}^{1/\alpha}\right)^{\alpha}\Bigg]\\
\end{eqnarray*}
For notational convenience, we denote $A=\sum_{j\in T}e_j^{1/\alpha}$ and $B=e_n^{1/\alpha}=e_{n+1}^{1/\alpha}$.
In order to prove the proposition it suffices to show that
\begin{eqnarray*}
(1-Y_n)A^{\alpha} + Y_n\cdot\left(A+B\right)^{\alpha}  \leq  (1-Y_n')\cdot(1-Y_{n+1}')\cdot A^{\alpha}+Y_n'\cdot(1-Y_{n+1}')\cdot\left(A+B\right)^{\alpha} \\
+(1-Y_n')\cdot Y_{n+1}'\cdot\left(A+B\right)^{\alpha}
+Y_n'\cdot Y_{n+1}'\cdot\left(A+2B\right)^{\alpha}
\end{eqnarray*}
Given the fact that $Y_n=Y_n'+Y_{n+1}'$, the above inequality can be rewritten equivalently as
\begin{equation*}
Y_n'Y_{n+1}'(A^\alpha-2(A+B)^{\alpha}+(A+2B)^{\alpha})\geq0
\end{equation*}
If either $Y_n'=0$ or $Y_{n+1}'=0$, the inequality is clearly true.
Otherwise, the inequality holds due to the convexity of the function $g(x)=x^{\alpha}$.
Specifically, for any $x,y>0$ and $\theta\in[0,1]$, it must be the case that
\begin{equation*}
g(\theta x+(1-\theta)y)\leq \theta g(x) +(1-\theta)g(y)
\end{equation*}
For $x=A$, $y=A+2B$ and $\theta=\frac{1}{2}$, we get that
\begin{equation*}
g\left(\frac{1}{2}A+\frac{1}{2}(A+2B)\right)\leq \frac{1}{2}g(A)+\frac{1}{2}g(A+2B) \Leftrightarrow 2(A+B)^{\alpha}\leq A^{\alpha}+(A+2B)^{\alpha}
\end{equation*}
\end{proof}

Proposition~\ref{prop:meanInequality} is a corollary of the generalized means inequality.

\begin{proposition}
\label{prop:meanInequality}
For any set of positive values $S=\{e_1,e_2,\ldots,e_n\}$ and constant $\alpha>1$,
\begin{equation*}
\left(\sum_{e_i\in S}e_i^{1/\alpha}\right)^{\alpha}\leq |S|^{\alpha-1}\sum_{e_i\in S}e_i
\end{equation*}
\end{proposition}

Proposition~\ref{TechnicalPoisson} estimates the moments of Binomial random variables through the moments of Poisson random variables.

\begin{proposition}\label{TechnicalPoisson}
For any $\alpha\geq 1$, the function $f(x)=x^{\alpha}$ and parameter $a\in [0,1]$ we have
$$\mathbb{E}[f(B_{a})] \leq \mathbb{E}[f(P_{a})]$$
where $B_{a}$ is a sum of $n$ independent Bernoulli random variables, $\mathbb{E}[ B_{a}]=a$ and $P_{a}$ is a Poisson random variable with parameter $a$.
\end{proposition}
\begin{proof}
To upper bound the expected value of $f(x)$, we will need the following probabilistic fact that was first proved by Hoeffding \cite{Hoeffding56}
for finite sum of Bernoulli random variables and was lately generalized for more general distributions by Berend and Tassa \cite{BerendT10}.

\begin{proposition}\emph{\cite{BerendT10}}\label{Berend}
Let $X=\sum_{i=1}^{t}X_i$ be the sum of $t$ (where $t$ is possibly equal to infinity) independent random variables,
$0\le X_i\le 1$ for $i=1,\dots, t$  and $\mu=\mathbb{E}[X]$.
For every convex function $f$,
$$\mathbb{E}[f(X)]\le \mathbb{E}[f(Y)]$$
where $Y$ is a binomial random variable with distribution $Y  \sim B(t,\mu/t)$
in case $t < \infty$, and a Poisson random variable with distribution $Y \sim P(\mu)$ otherwise.
\end{proposition}

We define a binomial random variable $B'_{a}$ as a sum of $B_{a}$ and an infinite number of Bernoulli random variables $Y'_j$ for $j=n+1,\dots,\infty$
such that $Pr(Y_j=1)=0$. Obviously, $\mathbb{E}[B'_{a}]=\mathbb{E}[B_{a}]=a$ and $\mathbb{E}[ f(B'_{a})]=\mathbb{E}[ f(B_{a})]$.
Since the function $f(x)$ is convex we can apply the Proposition \ref{Berend} with parameter $t=\infty$ and the statement follows.
\end{proof}

Proposition~\ref{prop:bounds} estimates moments of Poisson random variables with parameter $\lambda$
through the moments of Poisson random variables with parameter $1$.

\begin{proposition}\label{prop:bounds}
For any real $\alpha\geq 1$ and a Poisson random variable $P_{\lambda}$ with parameter $\lambda\geq 0$, we have:
\begin{itemize}
\item[(a)] If $0 \leq \lambda \leq 1$,
then $\mathbb{E}[P_{\lambda}^{\alpha}] \leq \lambda \mathbb{E}[P_{1}^{\alpha}]$.
\item[(b)] If $\lambda > 1$,
then $\mathbb{E}[P_{\lambda}^{\alpha}] \leq \lambda^{\alpha} \mathbb{E}[P_{1}^{\alpha}]$.
\end{itemize}
\end{proposition}
\begin{proof}
Recall that $\mathbb{E}[P_{\lambda}^{\alpha}]= \sum_{k=0}^{\infty}k^{\alpha}\frac{\lambda^k e^{-\lambda}}{k!}$.

\begin{itemize}
\item[(a)] Note that $e^{-(1-\lambda)}\geq 1-(1-\lambda)=\lambda\geq \lambda^{k-1}$ for $k\geq 2$ and $0 \leq \lambda \leq 1$.
Therefore, $e^{-1}\geq \lambda^{k-1}e^{-\lambda}$ for all $k\geq 2$.
For $\lambda=0$ the statement of the Lemma is trivial. Assume $\lambda>0$, then we derive
\begin{eqnarray*}
\mathbb{E}[P_{1}^{\alpha}]-\frac{1}{\lambda}\mathbb{E}[P_{\lambda}^{\alpha}]&=&\sum_{k=0}^{\infty}k^{\alpha}\frac{e^{-1}-\lambda^{k-1} e^{-\lambda}}{k!}\\
&=&(e^{-1}-e^{-\lambda})+\sum_{k=2}^{\infty}k^{\alpha}\frac{e^{-1}-\lambda^{k-1} e^{-\lambda}}{k!}\\
&\geq& (e^{-1}-e^{-\lambda})+\sum_{k=2}^{\infty}k\frac{e^{-1}-\lambda^{k-1} e^{-\lambda}}{k!}\\
&=&\sum_{k=0}^{\infty}k\frac{e^{-1}}{k!} -\frac{1}{\lambda}\sum_{k=0}^{\infty}k\frac{\lambda^{k} e^{-\lambda}}{k!}\\
&=&1-\frac{1}{\lambda}\cdot\lambda=0.
\end{eqnarray*}
\item[(b)] For the case where $\lambda > 1$, we will use two basic facts.
The first fact is that given a Poisson random variable $X_1$ with parameter $\lambda_1$ and a Poisson random variable $X_2$ with parameter $\lambda_2$
that are mutually independent, then a random variable $X_1+X_2$ is a Poisson random variable with parameter $\lambda_1+\lambda_2$.
The second fact is that for any random variable $X$ the quantity $\mathbb{E}[X^p]^{1/p}$ defines a norm and therefore satisfies the triangle inequality
(Minkowski's norm inequality), i.e. $||X+Y||_p\le ||X||_p+||Y||_p$.

Assume that $\lambda=\frac{A}{B}$ is a rational number, $A,B\in Z_+$ and $A>B\geq 1$.
Let $X$ be a Poisson random variable with parameter $\lambda$ and $Y_1,\dots, Y_A$ be independent Poisson random variables with parameter $1/B$.
In addition, let $Y_S$ be a random variable
\begin{equation*}
Y_S=\frac{\sum_{i\in S}Y_i}{{A-1 \choose B-1}}
\end{equation*}
Then
\begin{equation*}
X=\sum_{i=1}^AY_i=\sum_{S\subseteq [A]:|S|=B} Y_S
\end{equation*}
where $[A]=\{1,\dots,A\}$.
Let $P_1$ be a Poisson random variable with parameter $1$.
Then applying the triangle inequality we obtain
\begin{equation*}
\mathbb{E}[X^{\alpha}]^{1/\alpha}\leq \sum_{S\subseteq [A]:|S|=B} \mathbb{E}[Y_S^\alpha]^{1/\alpha}
=\frac{{A \choose B}}{{A-1 \choose B-1}}\mathbb{E}[P_1^\alpha]^{1/\alpha}
=\lambda \mathbb{E}[P_1^\alpha]^{1/\alpha}
\end{equation*}
which implies the inequality in the statement of the lemma for any rational value of $\lambda>1$.

To derive the inequality for any real $\lambda>1$ we just need to apply the standard limiting argument,
i.e., any real is a limit of rationals and the inequality holds for each of these rational values.
Therefore, the inequality must hold for the real value of $\lambda$.
\end{itemize}
\end{proof}

\section{Heterogeneous Multiprocessor without Migrations}
\label{section:non-migration}

In this section we consider the case where the migration of jobs is not permitted, but their preemption is allowed.
The corresponding homogeneous problem is known to be $\mathcal{NP}$-hard even if all jobs have common release dates and deadlines \cite{AlbersMS07}.
We propose an approximation algorithm by formulating the problem as a configuration integer program (IP) with an exponential number of variables
and a polynomial number of constraints.
Given an optimal solution for the configuration LP relaxation, we apply randomized rounding to get a feasible schedule for our problem.
In order to get a polynomial-time algorithm,
we present another (compact) formulation of our problem with a polynomial number of variables and constraints
and we show that the relaxations of the two formulations are equivalent.

\subsection{Linear Programming Relaxation}
\label{section:lpr}

In order to formulate our problem as a configuration IP we need to discretize the time.
In the following lemma we assume that the release dates and the deadlines of all jobs in all processors are integers.

\begin{lemma} \label{lemma:time-discrete}
There is a feasible schedule with energy consumption at most $((1+\frac{\varepsilon}{1-\varepsilon}) (1+\frac{2}{n-2}))^{\alpha} \cdot OPT$
in which each piece of each job $j \in \mathcal{J}$ ($j$ is executed on processor $i \in \mathcal{P}$)
starts and ends at a time point $r_{i,j}+ k\frac{\varepsilon}{n^3}(d_{i,j}-r_{i,j})$, where $k\geq0$ is an integer and $\varepsilon\in(0,1)$.
\end{lemma}
\begin{proof}

We will first transform an optimal schedule $\mathcal{S}^*$ to a feasible schedule $\mathcal{S}$ in which the execution time of each job $j \in \mathcal{J}$
executed on processor $i \in \mathcal{P}$ is at least $\frac{\varepsilon}{n}(d_{i,j}-r_{i,j})$.

As the release dates and the deadlines are integers, we can divide the time into unit length slots.
We can get now the schedule $\mathcal{S}$ from $\mathcal{S}^*$ as follows:
For each unit slot we increase the processors' speeds such that to create an idle period of length $\varepsilon$.
This can be done by increasing the speeds by a factor of $1+\frac{\varepsilon}{1-\varepsilon}$,
and hence the total energy consumption in $\mathcal{S}$ is increased by a factor of $(1+\frac{\varepsilon}{1-\varepsilon})^{\alpha}$.
For each job $j \in \mathcal{J}$,
we reserve an $\frac{\varepsilon}{n}$ period to each unit slot in $(r_{i,j},d_{i,j}]$ on the processor in which $j$ was executed in $\mathcal{S}^*$.
In $\mathcal{S}$, we decrease the speed of $j$ such that its total work to be executed during the periods where $j$ was executed in $\mathcal{S}^*$
and the additional $d_{i,j}-r_{i,j}$ reserved periods.
Therefore, in the final schedule the processing time of each job $j \in \mathcal{J}$ is at least $\frac{\varepsilon}{n}(d_{i,j}-r_{i,j})$.
After this transformation we apply the Earliest Deadline First (EDF) policy to each processor separately
with respect to the set of jobs assigned on this processor in $\mathcal{S}^*$ and the speeds defined above.
This ensures that we have a schedule with at most $n$ preemptions,
as in EDF a job may be interrupted only when another job is released.

Next, we transform $\mathcal{S}$ to a new schedule $\mathcal{S}'$ such that to satisfy the statement of the lemma.
For each job $j \in \mathcal{J}$ which is executed on the processor $i \in \mathcal{P}$,
we split the interval $(r_{i,j},d_{i,j}]$ into slots of length $\frac{\varepsilon}{n^3}(d_{i,j}-r_{i,j})$,
i.e., we partition $(r_{i,j},d_{i,j}]$ into intervals of the form
$(r_{i,j}+k\frac{\varepsilon}{n^3}(d_{i,j}-r_{i,j}),r_{i,j}+ (k+1)\frac{\varepsilon}{n^3}(d_{i,j}-r_{i,j})]$, where $k \geq 0$ is an integer.
As the processing time of $j$ in $\mathcal{S}$ is at least $\frac{\varepsilon}{n}(d_{i,j}-r_{i,j})$,
the execution of $j$ has been partitioned into at least $n^2$ slots.
In each of these slots, the job $j$ either is executed during the whole slot or is executed into a fraction of it.
As we have applied the EDF policy, the job $j$ is preempted at most $n$ times, and hence at most $2n$ of these slots are not completely covered by $j$,
since for each preempted piece of $j$ at most two slots may not be completely covered by it,
i.e., the first and the last slot of its execution.
We can modify the schedule $\mathcal{S}$ and get the schedule $\mathcal{S}'$
in which the job $j$ is executed only to the slots where it was entirely executed in $\mathcal{S}$.
The number of these slots is at least $n^2-2n$.
Thus, we have to increase the speed of $j$ by a factor of $1+\frac{2}{n-2}$,
and hence the energy is increased by a factor of $(1+\frac{2}{n-2})^{\alpha}$.
By taking into account that $\mathcal{S}$ is a factor of $(1+\frac{\varepsilon}{1-\varepsilon})^{\alpha}$ far from the optimal, the lemma follows.
\end{proof}

Let $\mathcal{S}$ be a schedule that satisfies Lemma~\ref{lemma:time-discrete}
and let $j\in \mathcal{J}$ be a job executed on the processor $i \in \mathcal{P}$ in $\mathcal{S}$.
The above lemma implies that the interval $(r_{i,j},d_{i,j}]$ can be partitioned into polynomial,
with respect to $n$ and $1/\varepsilon$, number of equal length slots.
In each of these slots, $j$ either is executed during the whole slot or is not executed at all.
In what follows we consider schedules that satisfy Lemma~\ref{lemma:time-discrete}.

A configuration $c$ is a schedule for a single job on a single processor.
Specifically, a configuration determines the slots, with respect to Lemma~\ref{lemma:time-discrete},
during which one job is executed.
Given a configuration $c$ for a job $j \in \mathcal{J}$, we can define the execution time of $j$
that is equal to the number of slots in $c$ multiplied by the length of the slot.
Due to the convexity of the speed-to-power function, in a minimum energy schedule that satisfies Lemma~\ref{lemma:time-discrete},
the job $j$ runs at a constant speed $s_j$.
Hence, $s_j$ is equal to the work of $j$ over its execution time.
Let $\mathcal{C}_{ij}$ be the set of all possible feasible configurations for a job $j \in \mathcal{J}$ in a processor $i \in \mathcal{P}$.

In order to ensure the feasibility of our schedule we need to further partition the time, by merging the slots for all jobs.
Given a processor $i \in \mathcal{P}$, consider the time points of all jobs of the form $r_{i,j}+k\frac{\varepsilon}{n^3}(d_{i,j}-r_{i,j})$
as introduced in Lemma~\ref{lemma:time-discrete}.
Let $t_{i,1}, t_{i,2},\ldots,t_{i,\ell_i}$ be the ordered sequence of these time points.
Consider now the intervals $(t_{i,p},t_{i,p+1}]$, $1 \leq p \leq \ell_i-1$.
In a schedule that satisfies Lemma~\ref{lemma:time-discrete},
in each such interval either there is exactly one job that is executed during the whole interval or the interval is idle.
Note also that these intervals might not have the same length.
Let $\mathcal{I}$ be the set of all these intervals for all processors.

We introduce the binary variable $x_{i,j,c}$ that is equal to one
if the job $j \in \mathcal{J}$ is entirely executed on the processor $i \in \mathcal{P}$ according to the configuration $c$, and zero otherwise.
Note that, given the configuration and the processor $i$ where the job $j$ is executed,
we can compute the energy consumption $E_{i,j,c}$ for the execution of $j$.
For ease of notation, we say $I \in (i,j,c)$ if the interval $I \in \mathcal{I}$ is included in the configuration $c$
of processor $i \in \mathcal{P}$ for the job $j \in \mathcal{J}$,
that is there is a slot $(r_{i,j}+k\frac{\varepsilon}{n^3}(d_{i,j}-r_{i,j}),r_{i,j}+(k+1)\frac{\varepsilon}{n^3}(d_{i,j}-r_{i,j})]$ in $c$ that contains $I$.
\begin{eqnarray}
\min \sum_{i,j,c} E_{i,j,c} \cdot x_{i,j,c} \nonumber\\
\sum_{i,c}x_{i,j,c}=1 & \hspace{2cm} \forall j\in\mathcal{J} \label{c3}\\
\sum_{(i,j,c)\ni I }x_{i,j,c}\leq 1 & \hspace{2cm} \forall I\in\mathcal{I} \label{c4}\\
x_{i,j,c}\in \{0,1\} & \hspace{2cm} \forall i\in\mathcal{P}, j\in\mathcal{J}, c\in\mathcal{C}_{ij} \label{c5}
\end{eqnarray}
Inequality~(\ref{c3}) enforces that each job is entirely executed according to exactly one configuration.
Inequality~(\ref{c4}) ensures that at most one job is executed in each interval $(t_{i,p},t_{i,p+1}]$, $1 \leq p \leq \ell_i-1$.

We next relax the constraints~(\ref{c5}) such that $x_{i,j,c}\geq0$.
Since the structure of this LP is quite simple we can define an equivalent compact LP relaxation with polynomial number of constraints and variables.
We describe how to do it in Section~\ref{section:compact}.
For now we assume that we can find an optimal solution of our configuration LP in polynomial time.

\subsection{Randomized Rounding}
\label{section:rr}

Now, we show how to apply randomized rounding to get an approximation algorithm for our problem.
Recall that, by definition, an interval $I\in\mathcal{I}$ corresponds to a single processor $i\in\mathcal{P}$.
Our algorithm follows.
\begin{algorithm}[h!] \nonumber
\caption{}
\label{algo}
\begin{algorithmic}[1]
\STATE Solve the configuration LP relaxation.
\STATE For each job $j \in \mathcal{J}$, choose a configuration at random with probability $x_{i,j,c}$.
\STATE Let $w_j(I)$ be the amount of work executed for job $j$ during the interval $I\in\mathcal{I}$ according to its chosen configuration and $i$ be the associated processor to the interval $I$.
\STATE Set $i$'s speed during $I$ as if $\sum_{j\in\mathcal{J}}w_j(I)$ units of work are executed with constant speed during the entire $I$.
\end{algorithmic}
\end{algorithm}

\begin{theorem}\label{thm:approx}
Assume that $\alpha_i\geq 1$ for all $i=1,\dots,m$.
Algorithm~\ref{algo} achieves an approximation ratio of $((1+\frac{\varepsilon}{1-\varepsilon}) (1+\frac{2}{n-2}))^{\alpha}\tilde{B}_{\alpha}$
for the heterogeneous multiprocessor preemptive speed-scaling problem without migrations in time polynomial to $n$ and to $1/\varepsilon$,
where $\alpha=\max_{i \in \mathcal{P}} \alpha_i$ and $\varepsilon\in(0,1)$.
\end{theorem}
\begin{proof}
For each interval $I \in \mathcal{I}$, we estimate the expected energy consumption during $I$.
So, in the remainder of the proof, we fix such an interval (and processor).

Initially, the algorithm computes an optimal solution for the relaxed LP.
For each job $j \in \mathcal{J}$, let $n_j$ be the number of the non-zero $x_{i,j,c}$ variables such that $I\in(i,j,c)$.
Note that, every such variable corresponds to some configuration $c$ such that if the job $j$ is executed according to $c$, then it must be executed during $I$.
For notational convenience, let $X_{j,k}$ be the $k$-th, $1 \leq k \leq n_j$, of these non-zero variables and $s_{j,k}$ be the corresponding speed.
The probability that the job $j$ is executed during $I$ in the algorithm's schedule is $Y_j=\sum_{k=1}^{n_j}X_{j,k}$.

If the job $j$ is entirely executed according to the configuration which corresponds to the variable $X_{j,k}$, then its energy consumption is $e_{j,k}=|I|s_{j,k}^{\alpha_i}$, during $I$.
The energy consumption achieved by the optimal
solution of the LP relaxation during $I$ is
$LP^*_I=\sum_{j=1}^n\sum_{k=1}^{n_j}e_{j,k}X_{j,k}$.

Assume that the
randomized rounding assigns exactly the jobs in the set $S$ to be
processed during the interval $I$. The probability of such an event
is $Pr(S)=\prod_{j\in S}Y_j\prod_{j\in\mathcal{J}\setminus
S}(1-Y_j)$. Let $E(S)$ be the expected energy consumption during $I$
under the condition that exactly the jobs in the set $S$ are
executed during $I$. Then, the expected energy consumption $E_I$ of
our algorithm during $I$ can be expressed as follows:
\begin{equation*}
E_I=\sum_{S\subseteq\mathcal{J}}E(S)\prod_{j\in S}Y_j\prod_{j\in\mathcal{J}\setminus S}(1-Y_j)
\end{equation*}

We, now, estimate $E(S)$.
Let $U(S)$ be the set of all combinations of pairs $(j,k)$ that we can choose in order to schedule exactly the jobs in set $S$, $S\subseteq\mathcal{J}$, during $I$.
If the algorithm schedules the jobs in $S$ during $I$ according to the configurations in $U$, where $U\subseteq U(S)$, then the total energy consumption during $I$ in the algorithm's schedule is $|I|\left(\sum_{(j,k)\in U}s_{j,k}\right)^{\alpha_i}$.
So,
\begin{equation*}
E(S)=\sum_{U\subseteq U(S)}\left(\prod_{(j,k)\in
U}\frac{X_{j,k}}{Y_j}\right)|I|\left(\sum_{(j,k)\in
U}s_{j,k}\right)^{\alpha_i}
\end{equation*}
For each job $j\in S$, we denote by $\tilde{e}_j$ a random variable which takes the value $e_{j,k}$ with probability $\frac{X_{j,k}}{Y_j}$.
Then, we have that
\begin{equation*}
E(S)=\mathbb{E}\left[\left(\sum_{j\in S}\tilde{e}_j^{1/\alpha_i}\right)^{\alpha_i}\right]
\end{equation*}
By Proposition \ref{Prop:ConcaveExpectation},
\begin{equation*}
E(S)\le \left(\sum_{j\in S}\mathbb{E}\left[\tilde{e}_j\right]^{1/\alpha_i}\right)^{\alpha_i}
\end{equation*}
We set $e_j=\mathbb{E}\left[\tilde{e}_j\right]$.
Therefore, the expected energy consumption during $I$ can be upper bounded as follows
\begin{equation*}
E_I\le \sum_{S\subseteq\mathcal{J}} \left(\sum_{j\in S
}e_j^{1/\alpha_i}\right)^{\alpha_i} \prod_{j\in
S}Y_j\prod_{j\in\mathcal{J}\setminus S}(1-Y_j)
\end{equation*}

We can assume that there exists a sufficiently large $Q \in \mathbb{N}$ such that $Y_j=\frac{q_j}{Q}$, $1 \leq j \leq n$, for some $q_j \in \mathbb{N}$
(we don't make any assumptions on the encoding length of these numbers, we use them only for analysis purposes) since these numbers come from solving an LP with rational coefficients.
Let $Y=1/Q$ and $q=\sum_{j=1}^nq_j$. Note that $q\le Q$. By applying the Proposition~\ref{prop:split} iteratively, we split   $Y_j$ into smaller pieces and we get that
\begin{eqnarray*}
E_I & \leq & \sum_{S\subseteq\{1,2,\ldots,q\}}
\left(\sum_{j\in S }e_j^{1/\alpha_i}\right)^{\alpha_i} Y^{|S|} (1-Y)^{q-|S|}\\
\end{eqnarray*}

By using Proposition \ref{prop:meanInequality} we get

\begin{eqnarray*}
E_I & \leq & \sum_{S\subseteq\{1,2,\ldots,q\}}|S|^{\alpha_i-1}\left(\sum_{j\in S }e_j\right) Y^{|S|} (1-Y)^{q-|S|}\\
 & = & \sum_{k=1}^q \sum_{S\subseteq\{1,2,\ldots,q\}, |S|=k}\left(\sum_{j\in S }e_j\right) k^{\alpha_i-1} Y^k (1-Y)^{q-k}\\
\end{eqnarray*}
By changing the order of the sums in the above inequality and given that ${q-1 \choose k-1}$ is the number of sets of cardinality $k$ that contain $j$, we get

\begin{eqnarray*}
E_I & \leq & \left(\sum_{j=1}^q e_{j}\right) \sum_{k=1}^q {q-1 \choose k-1} k^{\alpha_i-1} Y^k (1-Y)^{q-k}\\
  & = & \left(\frac{\sum_{j=1}^q e_{j}}{q}\right) \sum_{k=1}^q {q \choose k} k^{\alpha_i} Y^k (1-Y)^{q-k} \\
  & = & \left(\frac{\sum_{j=1}^n q_j e_{j}}{q}\right) \sum_{k=1}^q {q \choose k} k^{\alpha_i} Y^k (1-Y)^{q-k} \\
 & = & \frac{Q}{q} LP_I^* \sum_{k=1}^q {q \choose k} k^{\alpha_i} Y^k (1-Y)^{q-k} \\
 & = & \frac{Q}{q} LP_I^* \cdot  \mathbb{E}[B_{q/Q}^{\alpha_i}]
 \end{eqnarray*}
where   $B_{q/Q}$ is a random variable with expectation $\frac{q}{Q}$ which corresponds to the sum of $q$ i.i.d Bernoulli random variables.
Therefore,
\begin{eqnarray*}
E_I & \leq & \frac{Q}{q} LP_I^* \cdot \mathbb{E}[B_{q/Q}^{\alpha_i}]
 \leq \frac{Q}{q}LP_I^*\cdot \mathbb{E}[P_{q/Q}^{\alpha_i}]
 \leq \frac{Q}{q}LP_I^*\cdot \frac{q}{Q}\mathbb{E}[P_1^{\alpha_i}]
\end{eqnarray*}
where the second inequality follows from Proposition~\ref{TechnicalPoisson} and the last inequality follows from  Proposition~\ref{prop:bounds}(a).
Therefore, by summing over all intervals and processors and as $\alpha=\max_{i \in \mathcal{P}} \alpha_i$, we get
\begin{equation}
\nonumber
E\leq LP^*\cdot \mathbb{E}[P_1^{\alpha}]
=LP^*\cdot \tilde{B}_{\alpha}
\end{equation}
The theorem follows.
\end{proof}

\noindent{\bf Remark:}  Let $f_i$ be the contribution of all variables $x_{i,j,c}$ to the value of the optimal LP solution,
i.e., $\sum_{i \in \mathcal{P}} f_i=f$ where $f$ is the optimal value of the configuration LP.
One can refine the analysis of the Theorem \ref{thm:approx} and show the factor of $\sum_{i \in \mathcal{P}} \tilde{B}_{\alpha_i}\frac{f_i}{f}$
instead of $\tilde{B}_{\alpha}$ if all $\alpha_i\geq 1$.

\subsection{Compact Linear Programming Relaxation}
\label{section:compact}

Next, we define a compact formulation for the problem without migrations and
we show that the relaxations of the compact and the configuration LPs are equivalent.
Recall that, by Lemma \ref{lemma:time-discrete}, there is always an $((1+\frac{\varepsilon}{1-\varepsilon})(1+\frac{2}{n-2}))^{\alpha}$-approximate schedule
for our problem such that if the job $j \in \mathcal{J}$ is executed on the processor $i \in \mathcal{P}$, then its feasibility interval $(r_{i,j},d_{i,j}]$
can be partitioned into equal-length slots.
Given such a slot $t$, $j$ is either executed during the whole $t$ or it is not executed at all during $t$.
The number of these slots is $\frac{n^3}{\varepsilon}$,
while each slot $t$ has length $\ell_t=\frac{\varepsilon}{n^3}(d_{i,j}-r_{i,j})$.
Recall also that $\mathcal{I}$ denotes the set of all intervals occurred by merging the slots for all jobs.

In order to formulate our problem as a compact LP, we introduce a binary variable $y_{i,j,q}$
which is equal to one if the job $j$ is executed on the processor $i$ during exactly $q$ slots and zero otherwise.
Moreover, we introduce a binary variable $z_{i,j,q,t}$ which is equal to one if the job $j$ is executed on the processor $i$
during the slot $t$ and it is executed during exactly $q$ slots in total. Otherwise, $z_{i,j,q,t}$ is equal to zero.
We define the constants $p_{i,j,q}=q\frac{\varepsilon}{n^3}(d_{i,j}-r_{i,j})$
and $E_{i,j,q}=\frac{w_{i,j}^{\alpha_i}}{p_{i,j,q}^{\alpha_i-1}}$ as the total execution time and the energy consumption, respectively,
of the job $j$ if it is entirely executed on the processor $i$ during exactly $q$ slots.
\begin{eqnarray}
\min \sum_{i,j,q}E_{i,j,q}\cdot y_{i,j,q} \nonumber \\
\sum_{i,q}y_{i,j,q}=1 & \hspace{2cm} \forall j\in\mathcal{J} \label{c6}\\
\sum_t z_{i,j,q,t}=q\cdot y_{i,j,q} & \hspace{2cm} \forall i\in\mathcal{P},j\in\mathcal{J},q\in\{1,2,\ldots,\frac{n^3}{\varepsilon}\} \label{c7}\\
\sum_{j,q}\sum_{t:I\subseteq t}z_{i,j,q,t} \leq 1 & \hspace{2cm} \forall i\in\mathcal{P}, I\in\mathcal{I} \label{c8} \\
y_{i,j,q}, z_{i,j,q,t} \in \{0,1\}& \hspace{2cm} \forall i\in\mathcal{P}, j\in\mathcal{J}, q,t\in\{1,2,\ldots,\frac{n^3}{\varepsilon}\} \label{c10}
\end{eqnarray}
The constraint (\ref{c6}) ensures that each job is entirely executed on some processor.
The constraint (\ref{c7}) establishes the relationship between the variables $z_{i,j,q,t}$ and $y_{i,j,q}$.
If $y_{i,j,q}=1$, then exactly $q$ variables $z_{i,j,q,t}$ must be equal to one.
The constraint (\ref{c8}) enforces that at most one job is executed by each processor at each time.
Specifically, given a job $j \in \mathcal{J}$ which is executed on the processor $i \in \mathcal{P}$,
if $j$ is executed during the slot $t\in\{1,2,\ldots,\frac{n^3}{\varepsilon}\}$,
then $j$ is executed during every interval $I\in\mathcal{I}$ such that $I\subseteq t$.
Note that the numbers of both the variables and the constraints of the above LP are polynomial to $n$ and to $1/\varepsilon$.

The configuration and the compact formulations are equivalent,
as they both lead to a minimum energy schedule satisfying Lemma~\ref{lemma:time-discrete}.
Consider now the LPs that occur if we relax constraints~(\ref{c5}) and~(\ref{c10}), respectively.
In Lemma~\ref{lemma:equiv} we prove that the equivalence is also true for these relaxations through a transformation of a solution for the configuration LP relaxation
to a solution for the compact LP relaxation, and vice versa.
As a result, given a solution of the compact LP relaxation obtained by any polynomial time algorithm,
we can get a solution for the configuration LP relaxation.
Then, we can apply the randomized rounding presented in the previous section and get the approximation ratio of Theorem~\ref{thm:approx}.

\begin{lemma}\label{lemma:equiv}
The relaxations of the configuration LP and the compact LP are equivalent.
\end{lemma}
\begin{proof}
We will show that any feasible solution for the configuration LP relaxation
can be transformed to a feasible solution for the compact LP relaxation with the same energy consumption and vice versa.

Assume that we are given a feasible solution for the relaxation of the configuration LP.
Such a solution corresponds to a schedule of the jobs on the processors.
Specifically, the value of the variable $x_{i,j,c}$ specifies the part of the job $j \in \mathcal{J}$ executed on processor $i \in \mathcal{P}$
during the slots that belong to the configuration $c \in \mathcal{C}_{ij}$.
Let $\mathcal{C}_{ijq}\subseteq \mathcal{C}_{ij}$ be the set of configurations of $j$ on $i$ with exactly $q$ slots.
Then, we define  $z_{i,j,q,t}=\sum_{c\in \mathcal{C}_q:t\in c}x_{i,j,c}$.
This defines a feasible solution for the relaxation of the compact LP.

Assume that we are given a feasible solution for the compact LP.
We will define a set of configurations and we will assign a non-zero value for each variable $x_{i,j,c}$ that corresponds to these configurations.
The number of these configurations should be polynomial to $n$ and to $\frac{1}{\varepsilon}$.
The remaining variables of the configuration LP will be set to zero.

Consider a non-zero variable $y_{i,j,q}$ (and its corresponding variables $z_{i,j,q,t}$) in the solution of the compact LP.
We partition the part of the schedule defined by $y_{i,j,q}$ into a set of configurations with $q$ slots
and we specify the values of the variables $x_{i,j,c}$ that correspond to these configurations.
To do this, for each variable $y_{i,j,q}$ and its associated variables $z_{i,j,q,t}$, we construct a bipartite graph $G=(A\cup B,E)$ as follows.
The set $A$ contains $q$ nodes, i.e. $A=\{a_1,a_2,\ldots,a_q\}$.
Intuitively, each of these nodes corresponds to one of the $q$ slots of the configurations that will correspond to $y_{i,j,q}$.
The set $B$ contains $\frac{n^3}{\varepsilon}$ nodes, one for each possible slot of $j$ on the processor $i$ (see Lemma \ref{lemma:time-discrete}),
i.e. $B=\{b_1,b_2,\ldots,b_{\frac{n^3}{\varepsilon}}\}$.
We will define the set of edges $E$ and their weights, such that
each node $a_k\in A$ has weighted degree exactly $y_{i,j,q}$ and each node $b_t\in B$ has weighted degree exactly $z_{i,j,q,t}$.
Note that, the total weight of all the edges will be $q\cdot y_{i,j,q}=\sum_{t}z_{i,j,q,t}$.
We start by adding edges from $a_1$ to $b_1,b_2,\ldots$ of weight $z_{i,j,q,1},z_{i,j,q,2},\ldots$, respectively,
as long as $\sum_{t=1}^kz_{i,j,q,t}\leq y_{i,j,q}$.
The first time where $\sum_{t=1}^kz_{i,j,q,t}> y_{i,j,q}$ we add an edge between $a_1$ and $b_k$ of weight $y_{i,j,q}-\sum_{t=1}^{k-1}z_{i,j,q,t}$.
Moreover, we add an edge between $a_2$ and $b_k$ of weight $z_{i,j,q,k}-(y_{i,j,q}-\sum_{t=1}^{k-1}z_{i,j,q,t})$.
We continue adding edges from $a_2$ to $b_{k+1},b_{k+2},\ldots$ of weight $z_{i,j,q,k+1},z_{i,j,q,k+2},\ldots$, respectively,
until the sum of their weights is bigger than $y_{i,j,q}$.
At this point we add an edge of appropriate weight starting from $a_3$ and we continue like this.
Note that, by construction each node $b_t \in B$ has degree either one or two.
Then, we construct a set of configurations based on the following proposition.

\begin{proposition}
Let $G=(A \cup B,E)$ be a bipartite graph in which each node in $A$ has weighted degree exactly one and each node in $B$ has weighted degree at most one.
There are perfect matchings $M_1,M_2,\ldots,M_r$ (i.e., matchings having exactly $|A|$ edges) and coefficients $\lambda_1,\lambda_2,\ldots, \lambda_r$ such that $\sum_{i=1}^r\lambda_i=1$, and for each edge $e$ it holds that $\sum_{i:e\in M_i}\lambda_i=w_e$, where $w_e$ is the weight of the edge $e$.
\end{proposition}
\begin{proof}
By the construction of the graph $G$, all its edges have a positive weight and all nodes in the set $A$ have the same weighted degree which is equal to $y_{i,j,q}$.
Consider an arbitrary perfect matching $M$ in $G$.
Let $w_{\min} = \min_{e\in M}\{w_e\}$.
Clearly, $w_{\min} > 0$.
We define $\lambda_1 = w_{\min}$ and we modify the graph $G$ by setting the weight of every edge $e\in M$ equal to $w_e-w_{\min}$.
Then, we remove all edges with zero weight.
We repeat this procedure until the graph is empty.
Given that we remove at least one each in each iteration, we compute a polynomial number of perfect matchings.
\end{proof}

It is easy to see that the solution obtained for the configuration LP is feasible.
The fact that Constraint~(\ref{c3}) is satisfied comes from Constraints~(\ref{c6}) and~(\ref{c7}).
The fact that Constraint~(\ref{c4}) is satisfied comes from Constraint~(\ref{c8}).
\end{proof}

\section{Heterogeneous Multiprocessor with Migrations}
\label{section:migration}

In this section we present an algorithm for the heterogeneous multiprocessor speed-scaling problem with preemptions and migrations.
We assume that, if $x$ units of work for the job $j$ are executed on the processor $i$, then $x/w_{i,j}$ portion of $j$ is accomplished by $i$.
We formulate the problem as a configuration LP, with an exponential number of variables and a polynomial number of constraints,
and we show how to obtain an $OPT+\varepsilon$ solution with the Ellipsoid algorithm
in time polynomial to the size of the instance and to $1/\varepsilon$, where $\varepsilon>0$.

A configuration $c$ is a one-to-one assignment of $n_c$, $0 \leq n_c \leq m$,
jobs to the $m$ processors as well as an assignment of a speed value for every processor.
We denote by $\mathcal{C}$ the set of all possible configurations.
A well defined schedule for our problem has to specify exactly one configuration at each time $t$.
The cardinality of $\mathcal{C}$ is unbounded, since the processors' speeds may be any real values.
Hence, we have to discretize the possible speed values and consider only a finite number of speeds at which the processors can run.

\begin{lemma}\label{lemma:discr}
There is a feasible schedule of energy consumption at most $OPT+\varepsilon$ that uses
a finite (exponential to the size of the instance and polynomial to $1/\varepsilon$) number of discrete processors' speeds,
for any $\varepsilon>0$.
\end{lemma}
\begin{proof}
To discretize the speeds, we first define a lower and an upper bound on the speed of any processor in an optimal schedule.
For the lower bound, consider a job $j \in \mathcal{J}$.
Recall that the release date and the deadline of $j$ are different on different processors.
Hence, the feasible intervals of $j$ in different processors may be completely disjoint,
that is the processing time of $j$ in an optimal schedule can be equal to $\sum_{i \in \mathcal{P}} (d_{i,j}-r_{i,j})$.
Therefore, due to the convexity of the speed-to-power function,
a non-zero lower bound on the speed of every processor is the minimum density among all the jobs,
i.e., $s_{LB}=\min_{j \in \mathcal{J}} \{\frac{\min_{i \in \mathcal{P}}\{w_{i,j}\}}{\sum_{i \in \mathcal{P}} (d_{i,j}-r_{i,j})}\}$.
For the upper bound, consider a processor $i \in \mathcal{P}$.
An upper bound on the speed of $i$ can be obtained by calculating the speed at which the jobs would run if
they were all executed in the minimum alive interval of any job on $i$, i.e., $\frac{\sum_{j \in J}w_{i,j}}{\min_{j \in J}(d_{i,j}-r_{i,j})}$.
Hence, an upper bound on the speed of every processor is $s_{UB}=\max_{i \in P}\{\frac{\sum_{j \in J}w_{i,j}}{\min_{j \in J}(d_{i,j}-r_{i,j})}\}$.

Given these lower and upper bounds and a small constant $\delta>0$, we discretize the speed values in a geometric way.
In other words, we consider only the speeds of the form $(1+\delta) s_{LB}, (1+\delta)^2 s_{LB}, \ldots, (1+\delta)^k s_{LB}$,
where $k$ is the first integer such that $(1+\delta)^k s_{LB} \geq s_{UB}$.
Hence, the number of speed values is equal to $k=\left\lceil\log_{1+\delta} \frac{s_{UB}}{s_{LB}}\right \rceil$,
which is polynomial to the size of the instance and to $1/\log(1+\delta)$.

Consider now an optimal schedule for our problem.
Let $S$ be the schedule obtained from the optimal one by rounding up the processors' speeds to the closest discrete value.
The ratio of the energy consumption of any processor $i \in \mathcal{P}$ at any time $t$ in $S$
over the energy consumption by $i$ at $t$ in the optimal schedule is at most $(1+\delta)^{\alpha_i}$.
By summing up for all processors and all time instances, we get that the energy consumption of $S$ is at most $(1+\delta)^{\alpha} OPT$.
Finally, if we pick a $\delta$ such that $\delta=(1+\frac{\varepsilon}{OPT})^{1/\alpha}-1$,
then the energy consumption of $S$ is at most $OPT+\varepsilon$.
However, this selection made the number of discrete speeds to be exponential to the size of the instance and to $1/\varepsilon$.
\end{proof}

In what follows in this section, we deal with schedules that satisfy Lemma~\ref{lemma:discr}.
Let, now, $t_0<t_1<\ldots<t_{\ell}$ be the time instants that correspond to release dates and deadlines of jobs so that there is a time $t_i$
for every possible release date and deadline.
We denote by $\mathcal{I}$ the set of all possible intervals of the form $(t_{i-1},t_i]$, for $1\leq i\leq \ell$.
Let $|I|$ be the length of the interval $I$.

We introduce a variable $x_{I,c}$, for each $I\in\mathcal{I}$ and $c\in\mathcal{C}$,
which corresponds to the total processing time during the interval $I\in\mathcal{I}$ where the processors run according to the configuration $c\in\mathcal{C}$.
We denote by $E_{I,c}$ the instantaneous energy consumption of the processors if they run with respect to the configuration $c$ during the interval $I$.
Moreover, let $s_{j,c}$ be the speed of the job $j$ according to the configuration $c$.
For notational convenience, we denote by $(I,c)$ the set of jobs which are alive during the interval $I$ and
which are executed on some processor by the configuration $c$.
Finally, let $i(j,c)$ be the processor on which the job $j$ is assigned into configuration $c$.
We propose the following configuration LP:
\begin{eqnarray}
\min \sum_{I\in\mathcal{I},c\in\mathcal{C}}E_{I,c} \cdot x_{I,c} \nonumber\\
\sum_{c\in\mathcal{C}} x_{I,c} \leq |I| & \hspace{2cm} \forall I\in\mathcal{I} \label{c1}\\
\sum_{I,c:\;j\in(I,c)}\frac{s_{j,c}}{w_{i(j,c),j}}x_{I,c} \geq 1 & \hspace{2cm} \forall j\in\mathcal{J} \label{c2}\\
x_{I,c}\geq0 & \hspace{2cm} \forall I\in\mathcal{I},\; c\in\mathcal{C} \nonumber
\end{eqnarray}
Consider the schedule for the interval $I$ that occurs by an arbitrary order of the configurations assigned to $I$.
This schedule is feasible, as the processing time of all configurations assigned to $I$ is equal to the length of the interval.
Hence, Inequality~(\ref{c1}) ensures that for each interval $I$ there is exactly one configuration for each time $t \in I$.
Inequality~(\ref{c2}) implies that each job $j$ is entirely executed.

The above LP has an exponential number of variables.
In order to handle this, we create the dual LP, which has an exponential number of constraints.
Next, we show how to efficiently apply the Ellipsoid algorithm to it (see \cite{GrotschelLS93}).
For this, we provide a polynomial-time separation oracle, i.e.,
we give a polynomial-time algorithm which given a solution for the dual LP
decides if this solution is feasible or otherwise it identifies a violated constraint.
As we can compute an optimal solution for the dual LP,
we can also find an optimal solution for the primal LP by solving it with the variables
corresponding to the constraints that were found to be violated during the run of the ellipsoid
method and setting all other primal variables to be zero.
The number of these violated constraints is polynomial to the size of the instance and to $1/\varepsilon$.
Thus, we can solve the primal LP with a polynomial number of variables.

The dual LP is the following:
\begin{eqnarray*}
\max \; \sum_{j\in\mathcal{J}}\lambda_j- \sum_{I\in\mathcal{I}}\mu_I|I|\\
\sum_{j\in(I,c)}\frac{s_{j,c}}{w_{i(j,c),j}}\lambda_j-\mu_I\leq E_{I,c} & \hspace{2cm} \forall I\in\mathcal{I},\;c\in\mathcal{C} \\
\mu_I,\lambda_j\geq0 & \hspace{2cm} \forall I\in\mathcal{I},\;j\in\mathcal{J}
\end{eqnarray*}

The separation oracle for the dual LP works as follows.
For each $I\in\mathcal{I}$, we try to find if there is a violated constraint.
Recall that there are $O(nm)$ intervals in the set $\mathcal{I}$.
For a given $I$, it suffices to check the minimum among the values $E_{I,c}-\sum_{j\in(I,c)}\frac{s_{j,c}}{w_{i(j,c),j}}\lambda_j$
over all possible configurations $c$.
If this minimum value is less than $-\mu_I$, then we have a violated constraint.
Otherwise, if we cannot find any violated constraint for all $I \in \mathcal{I}$, then the dual solution is feasible.

Note here that $E_{I,c}=\sum_{j\in(I,c)}s_{j,c}^{\alpha_{i(j,c)}}$, and hence we want to find the minimum value of
$\sum_{j\in(I,c)}(s_{j,c}^{\alpha_{i(j,c)}}-\frac{s_{j,c}}{w_{i(j,c),j}}\lambda_j)$.
For each job $j \in \mathcal{J}$ that is alive during $I$, the term $s_{j,c}^{\alpha_{i(j,c)}}-\frac{s_{j,c}}{w_{i(j,c),j}}\lambda_j$
is minimized at the discrete value $v_{i(j,c),j}$ which is one of the two closest possible discrete speeds to the value
$\left(\frac{\lambda_j}{\alpha_{i(j,c)}\cdot w_{i(j,c),j}}\right)^{1/(\alpha_{i(j,c)}-1)}$.
To see this we just need to notice that we minimize a one variable convex function over a set of possible discrete values.
The value $\left(\frac{\lambda_j}{\alpha_{i(j,c)}\cdot w_{i(j,c),j}}\right)^{1/(\alpha_{i(j,c)}-1)}$ is obtained by minimizing
$s_{j,c}^{\alpha_{i(j,c)}}-\frac{s_{j,c}}{w_{i(j,c),j}}\lambda_j$ if there is no discretization
of the speeds and it is obtained by equating the derivative of the last expression with zero.
Hence, given an interval $I$, we want to find a configuration $c$ that minimizes
$\sum_{j\in(I,c)}(v_{i(j,c),j}^{\alpha_{i(j,c)}}-\frac{v_{i(j,c),j}}{w_{i(j,c),j}}\lambda_j)$.

Since a configuration $c$ assigns $0\leq n_c \leq m$ jobs to $m$ processors,
the problem of minimizing the last expression reduces to a maximum weighted matching
on the bipartite graph which is constructed as follows: we introduce one node for each job and one node for each processor.
There is an edge between each alive job $j \in \mathcal{J}$ in interval $I$ and each processor $i \in \mathcal{P}$
with weight equal to $-(v_{i,j}^{\alpha_i}-\frac{v_{i,j}}{w_{i,j}}\lambda_j)$.
The maximum weight matching in such a bipartite graph defines  a configuration $c$, that is an assignment of $n_c\leq m$ jobs to $m$ processors.

Hence, there is a polynomial time separation oracle for the dual problem which runs in polynomial time.
To apply the ellipsoid method in polynomial time, we need to check two additional technical conditions.
The first condition is that the value of all dual variables are upper bounded by some number $R$.
The second condition is that for the dual program there is a feasible point (or solution) and every point in a radius $r$ is feasible.
Then the running time of the ellipsoid method will be polynomial to $\log \frac{R}{r}$.

The first condition and the bound on $R$ can be derived from the fact that the solution of the problem must be a vertex of the
corresponding polyhedra since we know that the value of an optimal solution is bounded.
Therefore, $R$ is a polynomial involving various input parameters.
We skip the precise definition of $R$.
The second condition is satisfied for the point $(\lambda,\mu)$
defined as follows: $\lambda_j=1$ for all $j \in \mathcal{J}$ and $\mu_I$ is large enough such that
$-\mu_I+1\le \min_c \left(E_{I,c}-2\sum_{j\in(I,c)}\frac{s_{j,c}}{w_{i(j,c),j}}\right)$.
Hence, the inequalities are satisfied in the ball of radius 1 around $(\lambda,\mu)$, that is $r=1$.

\begin{theorem}
A schedule for the heterogeneous multiprocessor speed-scaling problem with migrations of energy consumption $OPT+\varepsilon$
can be found in polynomial time with respect to the size of the instance and to $1/\varepsilon$, for any $\varepsilon>0$.
\end{theorem}

%%%%%%%%%%%%%%%%%%%%%%%%%%%%%%%%%%%%%%%%%%%%%%%%%%%%%%%%%%%%%%%%%%%%%%%%%%%%%%%%%%%%%%%%%

\section{Single processor without Preemptions}
\label{section:single}

In this section we present an approximation algorithm for the single processor speed-scaling problem in which the preemption of jobs is not allowed.
As a single processor is available, each job $j \in \mathcal{J}$ has a unique release date $r_j$, deadline $d_j$ and amount of work $w_j$,
while when the processor runs at speed $s$, it consumes energy with rate $s^{\alpha}$.
Due to the convexity of the speed-to-power function, $j$ runs at a constant speed $s_j$ in an optimal schedule $\mathcal{S}^*$.
Antoniadis and Huang \cite{AntoniadisH12} proved that this problem is $\mathcal{NP}$-hard and gave a $2^{5\alpha-4}$-approximation algorithm.

The algorithm in \cite{AntoniadisH12} consists of a series of transformations of the initial instance.
Our algorithm applies the first of these transformations.
Then, we give a transformation to the heterogeneous multiprocessor speed-scaling problem without migrations.

For completeness, we describe the first transformation given in \cite{AntoniadisH12}.
We partition the time as follows:
let $t_1$ be the smallest deadline of any job in $\mathcal{J}$, i.e., $t_1=\min\{d_j : j \in \mathcal{J}\}$.
Let $\mathcal{J}_1 \subseteq \mathcal{J}$ be the set of jobs which are released before $t_1$, i.e., $\mathcal{J}_1=\{j \in \mathcal{J} : r_j \leq t_1\}$.
Next, we set $t_2=\min\{d_j : j \in \mathcal{J} \setminus \mathcal{J}_1\}$ and $\mathcal{J}_2=\{j \in \mathcal{J} : t_1 < r_j \leq t_2\}$,
and we continue this procedure until all jobs are assigned into a subset of jobs.
Let $k$ be the number of subsets of jobs that have been created.
Moreover, let $t_0=\min\{r_j: j \in \mathcal{J}\}$ and $t_{k+1}=\max\{d_j: j \in \mathcal{J}\}$.

Consider the intervals $(t_{i-1},t_i]$, $1 \leq i \leq k+1$.
Let $\mathcal{I}_j$ be the set of intervals in which the job $j \in \mathcal{J}$ is alive.
In some of them $j$ is alive during the whole interval, while in at most two of them it is alive during a part of the interval.
Consider now the non-preemptive problem in which the execution of $j$ should take place into exactly one interval $I \in \mathcal{I}_j$.
Note that the execution of $j$ should respect its release date and its deadline.

\begin{proposition} \label{prop:ah}
Let $\mathcal{S}$ be an optimal non-preemptive schedule for the problem in which the execution of each job $j \in \mathcal{J}$
should take place into exactly one interval $I \in \mathcal{I}_j$.
It holds that $E(\mathcal{S}) \leq 2^{\alpha-1} OPT$.
\end{proposition}
\begin{proof}
In order to get a relation about the energy consumption between the schedule $\mathcal{S}$ and the optimal schedule $\mathcal{S}^*$,
consider first a job $j \in \mathcal{J}_{\ell}$ which is alive in more than one intervals, i.e., $|\mathcal{I}_j|\geq2$.
By definition, it holds that $r_j \leq t_{\ell}$ and $t_{\ell'} < d_j \leq t_{\ell'+1}$, where $\ell < \ell'$.
Moreover, consider a $p$, $\ell < p \leq \ell'$, and let $j' \in \mathcal{J}_p$ be the job that defines $t_p$, i.e., $d_{j'}=t_p$.
By definition, for $j'$ it holds that $t_{p-1} < r_{j'} \leq t_p$.
Although $j$ is alive at times $t_{p-1}$ and $t_p$, there is no feasible schedule in which $j$ is executed at both of them;
otherwise $j'$ could not be feasibly executed as we have available only one processor.
Therefore, in $\mathcal{S}^*$ a job cannot appear into more than two consecutive intervals $(t_{\ell-1},t_{\ell}]$ and $(t_{\ell},t_{\ell+1}]$.

Starting from $\mathcal{S}^*$, we create a feasible non-preemptive schedule $\mathcal{S}'$ for the problem
in which the execution of each job $j \in \mathcal{J}$ should take place into exactly one interval $I \in \mathcal{I}_j$
respecting its release date and its deadline.
In order to do this, consider a job $j \in \mathcal{J}$ which is executed into two intervals in $\mathcal{S}^*$,
let $(t_{\ell-1},t_{\ell}]$ and $(t_{\ell},t_{\ell+1}]$.
Let $e_{j,\ell}$ and $e_{j,\ell+1}$ be the execution time of $j$ into $(t_{\ell-1},t_{\ell}]$ and $(t_{\ell},t_{\ell+1}]$, respectively.
Assume, w.l.o.g., that $e_{j,\ell} \geq e_{j,\ell+1}$.
In $\mathcal{S}$, we execute the whole work of $j$ during $(t_{\ell-1},t_{\ell}]$ such that its execution takes exactly $\frac{(e_{j,\ell} + e_{j,\ell+1})}{2}$ time.
In order to do this, we just have to increase the speed $s_j$ that $j$ had in $\mathcal{S}^*$ by a factor of $2$.
Hence, the energy consumption of $j$ in $\mathcal{S}^*$ was $(e_{j,\ell} + e_{j,\ell+1})s_j^{\alpha}$,
while in $\mathcal{S}'$ is $\frac{(e_{j,\ell} + e_{j,\ell+1})}{2}(2s_j)^{\alpha}$.
By summing up for all jobs we get that $E(\mathcal{S}') \leq 2^{\alpha-1} OPT$.
As $\mathcal{S}$ is an optimal schedule, we get that $E(\mathcal{S}) \leq 2^{\alpha-1} OPT$.
\end{proof}

Next, we describe how to pass from the transformed problem to the heterogeneous multiprocessor speed-scaling problem without migrations.
For every interval $(t_{i-1},t_i]$, $1 \leq i \leq k+1$, we create a processor $i$.
For every job $j \in \mathcal{J}$ which is alive during a part or during the whole interval $(t_{i-1},t_i]$, $1 \leq i \leq k+1$, we set:
(i) $r_{i,j}=0$ if $r_j \leq t_{i-1}$ or $r_{i,j}=r_j-t_{i-1}$ if $r_j > t_{i-1}$,
(ii) $d_{i,j}=t_i-t_{i-1}$ if $d_j > t_i$ or $r_{i,j}=d_j-t_{i-1}$ if $d_j \leq t_i$, and
(iii) $w_{i,j}=w_j$.
For each processor $i$, $1 \leq i \leq k+1$, we set $\alpha_i=\alpha$.

We next apply the approximation algorithm presented in Section~\ref{section:non-migration} which is based on the rounding of a configuration LP.
Note that the number of configurations of each job here is polynomial to $n$ and to $1/\varepsilon$,
as we consider that preemptions are not allowed and hence a configuration can only contain continuous slots.
Thus, the resulting LP after the transformation has polynomial size and it can be directly solved
without using the compact LP presented in Section~\ref{section:compact}.

Note also that the algorithm presented in Section~\ref{section:non-migration} will create a preemptive schedule $\mathcal{S}$.
However, we can transform $\mathcal{S}$ into a non-preemptive schedule $\mathcal{S}'$ of the same energy consumption.
To see this, note that in each processor $i$, $1 \leq i \leq k+1$,
each job $j \in \mathcal{J}$ has $r_{i,j}=0$ or $d_{i,j}=t_i-t_{i-1}$.
Hence, by applying the Earliest Deadline First policy to each processor separately we can get the non-preemptive schedule $\mathcal{S}'$.

\begin{theorem}
The single processor speed-scaling problem without preemptions can be approximated within a factor of
$2^{\alpha-1} ((1+\frac{\varepsilon}{1-\varepsilon}) (1+\frac{2}{n-2}))^{\alpha}\tilde{B}_{\alpha}$,
where $\varepsilon \in (0,1)$.
\end{theorem}

\section{Job Shop Scheduling with Preemptions}
\label{section:jobshop}

In this section, we consider the energy minimization problem in a job shop environment.
An instance of the problem contains a set of jobs $\mathcal{J}$,
where each job $j\in\mathcal{J}$ consists of $\mu_j$ operations $O_{j,1},O_{j,2},\ldots,O_{j,\mu_j}$, which must be executed in this order.
That is,
there are precedence constraints of the form $O_{j,k}\rightarrow O_{j,k+1}$, for each $j\in \mathcal{J}$ and $1\leq k\leq \mu_j-1$,
meaning that the operation
$O_{j,k+1}$ can start only once the operation $O_{j,k}$ has finished.
Let $\mu$ be the number of all the operations, i.e. $\mu=\sum_{j\in\mathcal{J}}\mu_j$.
Each operation $O_{j,k}$ has an amount of work $w_{j,k}$.
Moreover, we are given a set of $m$ heterogeneous processors $\mathcal{P}$.
Every operation $O_{j,k}$, $j\in\mathcal{J}$ and $1\leq k\leq \mu_j$, is also associated with a single processor $i\in\mathcal{P}$ on which it must be entirely executed.
Note that more than one operations of the same job may have to be executed on the same processor.
Furthermore, for each operation $O_{j,k}$, we are given a release date $r_{j,k}$ and a deadline $d_{j,k}$.
For each $j \in \mathcal{J}$, we can assume that $r_{j,1} \leq r_{j,2} \leq \ldots \leq r_{j,\mu_j}$
as well as $d_{j,1} \leq d_{j,2} \leq \ldots \leq d_{j,\mu_j}$.
Preemptions of operations are allowed.
The objective is to find a feasible schedule of minimum energy consumption.

Next, we formulate the job shop problem as an integer configuration LP.
A configuration is a schedule for a job, i.e. a schedule for all its operations.
In order to define formally the notion of a configuration, we have to discretize the time.
We define the time points $t_0,t_1,\ldots,t_{\tau}$, in increasing order, where each $t_{\ell}$ corresponds to either
a release date or a deadline, so that there is a corresponding $t_{\ell}$ for each possible release date and deadline of an operation.
Then, we define the intervals $I_{\ell}=(t_{\ell-1},t_{\ell}]$, for $1\leq \ell \leq \tau$, and we denote by $|I_{\ell}|$ the length of $I_{\ell}$.
We further discretize the time inside each interval $I_{\ell}$, $1 \leq \ell \leq \tau$, based on the following lemmas in which it is assumed that the release dates and the deadlines of all operations are integers.

\begin{lemma} \label{lemma:discjs}
There is a feasible schedule with energy consumption at most $(1+\varepsilon)^{\alpha} \cdot OPT$
in which each piece of each operation $O_{j,k}$, $j\in \mathcal{J}$ and $1 \leq k \leq \mu_j$, executed during the interval $I_{\ell}$, $1 \leq \ell \leq \tau$,
starts and ends at a time point $t_{\ell-1}+h\frac{\varepsilon}{\mu(1+\varepsilon)}|I_{\ell}|$, where $h\geq 0$ is an integer and $\varepsilon \in (0,1)$.
\end{lemma}
\begin{proof}
Consider an optimal schedule $\mathcal{S}^*$ for our problem and an interval $I_{\ell}$, $1 \leq \ell \leq \tau$.
We define the time points $u_0=t_{\ell-1},u_1,u_2,\ldots,u_p=t_{\ell}$, in increasing order,
where each $u_q$, $0 \leq q \leq p$, corresponds to either a begin time or a completion time
of a piece of an operation on any processor during $I_{\ell}$ in $\mathcal{S}^*$,
so that for each begin time and completion time there is a corresponding $u_q$.
We call the interval $(u_{q-1},u_q]$, for $1 \leq q \leq p$, a \emph{slice}.
Consider any such slice and any processor $i\in\mathcal{P}$.
During the whole slice, the processor is either idle of fully occupied by a single operation.

Note that we can see the part of the schedule $\mathcal{S}^*$ during the interval $I_{\ell}$
as a schedule for the preemptive job shop problem without speed-scaling where the makespan is at most $|I_{\ell}|$.
Baptiste \cite{BaptisteCKQSS11} at al. (see their Corollary 4.2) showed that there is always a schedule for this problem with at most $\mu$ slices.

We will now transform $\mathcal{S}^*$ to a feasible schedule $\mathcal{S}$ satisfying the lemma.
Consider an interval $I_{\ell}$, $1 \leq \ell \leq \tau$.
We first create an idle period of length at least $\frac{\varepsilon}{1+\varepsilon}|I_{\ell}|$.
This can be done by increasing the speeds of all processors of all slices in $I_{\ell}$ by a factor of $1+\varepsilon$.
Hence, the energy consumption becomes at most a factor of $(1+\varepsilon)^{\alpha}$ far from the energy of $\mathcal{S}^*$.
In order to obtain $\mathcal{S}$, we round up the length of each slice to the closest $h\frac{\varepsilon}{\mu(1+\varepsilon)}|I_{\ell}|$.
In this way, the length of each slice is increased by at most $\frac{\varepsilon}{\mu(1+\varepsilon)}|I_{\ell}|$.
Since the number of slices is at most $\mu$, the total processing time in $I_{\ell}$ is increased by at most
$\mu(\frac{\varepsilon}{\mu(1+\varepsilon)}|I_{\ell}|)=\frac{\varepsilon}{1+\varepsilon}|I_{\ell}|$,
which is the length of the created idle period.
Thus, $\mathcal{S}$ is a feasible schedule, and the lemma follows.
\end{proof}

\begin{lemma} \label{lemma:discjs2}
There is a feasible schedule with energy consumption at most
$(1+\varepsilon)^{\alpha} (1+\frac{2}{\mu-2})^{\alpha} (1+\frac{\varepsilon}{1-\varepsilon})^{\alpha} \cdot OPT$ such that,
for each operation $O_{j,k}$, $j\in \mathcal{J}$ and $1 \leq k \leq \mu_j$,
there are two time points $b_{j,k}$ and $c_{j,k}$, as the ones defined in Lemma~\ref{lemma:discjs}, so that
each piece of $O_{j,k}$ starts and ends at a time point $b_{j,k}+h\frac{\varepsilon}{\mu^3}(c_{j,k}-b_{j,k})$ in $(b_{j,k},c_{j,k}]$,
where $h\geq0$ is an integer and $\varepsilon\in(0,1)$.
\end{lemma}
\begin{proof}
Consider a schedule $\mathcal{S}$ satisfying Lemma~\ref{lemma:discjs}.
In $\mathcal{S}$, each interval $I_{\ell}$, $1 \leq \ell \leq \tau$, is partitioned into polynomial to $\mu$ and to $1/\varepsilon$ number of equal length slots.
In each of these slots, each operation $O_{j,k}$, $j\in \mathcal{J}$ and $1 \leq k \leq \mu_j$, is either executed during the whole slot or is not executed at all.
Let $b_{j,k}$ and $c_{j,k}$ be the starting time of the first piece and the completion time of the last piece, respectively, of $O_{j,k}$ in $\mathcal{S}$.

We will first transform the schedule $\mathcal{S}$ to a feasible schedule $\mathcal{S}'$ in which the execution time of each operation
$O_{j,k}$, $j\in \mathcal{J}$ and $1 \leq k \leq \mu_j$, is at least $\frac{\varepsilon}{\mu}(c_{j,k}-b_{j,k})$.
For each time slot $s$ of Lemma~\ref{lemma:discjs} we increase the processors' speeds in order to create an idle period of length $\varepsilon|s|$,
where $|s|$ is the length of the slot.
This can be done by increasing the speeds by a factor of $1+\frac{\varepsilon}{1-\varepsilon}$,
and hence the total energy consumption in $\mathcal{S}$ is increased by a factor of $(1+\frac{\varepsilon}{1-\varepsilon})^{\alpha}$.
For each operation $O_{j,k}$, $j\in \mathcal{J}$ and $1 \leq k \leq \mu_j$,
we reserve an $\frac{\varepsilon|s|}{\mu}$ period to each slot $s$ in $(b_{j,k},c_{j,k}]$.
We then decrease the speed of $O_{j,k}$ so that its total work is executed during the periods where $O_{j,k}$ was executed in $\mathcal{S}$
and the additional $c_{j,k}-b_{j,k}$ reserved periods.
Therefore, in the final schedule the processing time of each operation $O_{j,k}$ is at least $\frac{\varepsilon}{\mu}(c_{j,k}-b_{j,k})$.
After this transformation we apply the Earliest Deadline First (EDF) policy to the operations of each processor separately,
considering as release date and deadline of each operation $O_{j,k}$, $j\in \mathcal{J}$ and $1 \leq k \leq \mu_j$,
the time points $b_{j,k}$ and $c_{j,k}$, respectively.
This ensures that we have a feasible schedule with at most $\mu$ preemptions,
as in EDF an operation may be interrupted only when another operation is released.

Next, we transform $\mathcal{S}'$ to a new schedule $\mathcal{S}''$ satisfying the lemma.
For each operation  $O_{j,k}$, $j\in \mathcal{J}$ and $1 \leq k \leq \mu_j$,
we split the interval $(b_{j,k},c_{j,k}]$ into slots of length $\frac{\varepsilon}{\mu^3}(c_{j,k}-b_{j,k})$,
i.e., we partition $(b_{j,k},c_{j,k}]$ into intervals of the form
$(b_{j,k}+h\frac{\varepsilon}{\mu^3}(c_{j,k}-b_{j,k}),b_{j,k}+ (h+1)\frac{\varepsilon}{\mu^3}(c_{j,k}-b_{j,k})]$, where $h \geq 0$ is an integer.
As the processing time of $j$ in $\mathcal{S}$ is at least $\frac{\varepsilon}{\mu}(c_{j,k}-b_{j,k})$,
the execution of $O_{j,k}$ has been partitioned into at least $\mu^2$ slots.
In each of these slots, the operation $O_{j,k}$ either is executed during the whole slot or is executed into a fraction of it.
As we have applied the EDF policy, each operation is preempted at most $\mu$ times.
Thus, among the time slots that $O_{j,k}$ is executed, at most $2\mu$ of them are not fully occupied by $O_{j,k}$ because for each preempted piece of $O_{j,k}$ at most two slots may not be completely covered by it.
We can modify the schedule $\mathcal{S}'$ and get the schedule $\mathcal{S}''$
in which the operation $O_{j,k}$ is executed only to the slots where it was entirely executed in $\mathcal{S}'$.
The number of these slots is at least $\mu^2-2\mu$.
Thus, we have to increase the speed of $O_{j,k}$ by a factor of $1+\frac{2}{\mu-2}$,
and hence the energy is increased by a factor of $(1+\frac{2}{\mu-2})^{\alpha}$.
By taking into account Lemma~\ref{lemma:discjs}
and the fact that $\mathcal{S}'$ is a factor of $(1+\frac{\varepsilon}{1-\varepsilon})^{\alpha}$ far from $\mathcal{S}$, the lemma follows.
\end{proof}

Henceforth, we consider schedules that satisfy the above lemma.
That is, for each operation $O_{j,k}$, we consider that there is a polynomial number of candidate time points $b_{j,k}$ and $c_{j,k}$ such that $O_{j,k}$ is entirely executed during $(b_{j,k},c_{j,k}]$.
Moreover, the interval $(b_{j,k},c_{j,k}]$ is partitioned into a polynomial number of equal length slots so that, given such a slot, the operation $O_{j,k}$ is either executed during the whole slot or is not executed at all during that slot.

Now, we can formulate our problem as an integer program.
A configuration $c$ is a schedule for a single job $j$, i.e. a feasible schedule for all its operations.
So, a configuration specifies the interval $(b_{j,k},c_{j,k}]$ and the time slots inside this interval, with respect to Lemma~\ref{lemma:discjs2}, during which each operation $O_{j,k}$ of the job $j$ is executed.
Let $\mathcal{C}_j$ be the set of all possible feasible configurations for job $j\in \mathcal{J}$.

In order to proceed, we need an additional definition by combining the slots of all the operations.
Specifically, given a processor $i \in \mathcal{P}$, consider the time points of all operations of the form
$b_{j,k}+h\frac{\varepsilon}{\mu^3}(c_{j,k}-b_{j,k})$ as introduced in Lemmas~\ref{lemma:discjs} and~\ref{lemma:discjs2}.
Let $t_{i,1}, t_{i,2},\ldots,t_{i,p_i}$ be the ordered sequence of these time points on the processor $i\in\mathcal{P}$.
In a schedule that satisfies Lemma \ref{lemma:discjs2},
in each interval $(t_{i,q},t_{i,q+1}]$, $1 \leq q \leq p_i-1$, either there is exactly one operation that is executed during the whole interval or the interval is idle on $i$.
Note also that these intervals might not have the same length.
Let $\mathcal{I}$ be the set of all these intervals for all processors.
According to Lemmas~\ref{lemma:discjs} and~\ref{lemma:discjs2},
the size of $\mathcal{I}$ is polynomial to the size of the instance and to $1/\varepsilon$.

Recall that the execution interval of a job $j$ can be partitioned into a set of equal length time slots so that, for every such time slot, either a single operation of $j$ is executed during the whole slot or $j$ is not executed at all during that slot.
It has to be noticed that, by definition, every such slot consists of one or more intervals in $\mathcal{I}$ and every interval in $\mathcal{I}$ (during which some operation of $j$ is alive) is contained entirely in a single slot of $j$.

If we know the configuration according to which the job $j$ is executed,
we can compute the energy consumption $E_{j,c}$ for the execution of $j$ because there is always an optimal schedule such that each operation is executed with constant speed.
For notational convenience, we say that $I\in(j,c)$, if the job $j$ is executed during the interval $I \in \mathcal{I}$ according to the configuration $c$.
That is, there is an operation $O_{j,k}$, two time points $b_{j,k}$ and $c_{j,k}$, and a slot
$(b_{j,k}+h\frac{\varepsilon}{\mu^3}(c_{j,k}-b_{j,k}),b_{j,k}+(h+1)\frac{\varepsilon}{\mu^3}(c_{j,k}-b_{j,k})]$
in $c$ that contains $I$.
\begin{eqnarray}
\min \sum_{j,c} E_{j,c} \cdot x_{j,c} \nonumber\\
\sum_{c}x_{j,c}\geq 1 & \hspace{2cm} \forall j\in\mathcal{J} \label{js1}\\
\sum_{c\in\mathcal{C}_j : I\in(j,c)}x_{j,c}\leq 1 & \hspace{2cm} \forall I \in \mathcal{I} \label{js2}\\
x_{j,c}\in \{0,1\} & \hspace{2cm} \forall j\in\mathcal{J}, c\in\mathcal{C}_j \label{js3}
\end{eqnarray}

Constraint~(\ref{js1}) enforces that each job is entirely executed according to exactly one configuration.
Constraint~(\ref{js2}) ensures that at most one job is executed in each interval $I \in \mathcal{I}$.
We consider the relaxed LP of the above integer program where the integrality constraints $x_{j,c}\in \{0,1\}$
are replaced by the constraints $x_{j,c}\geq0$, for all $j\in\mathcal{J}$ and $c\in\mathcal{C}_j$.
This LP contains an exponential number of variables but it can be solved in polynomial time by applying the Ellipsoid algorithm to its dual as we explain in the following.
The dual LP is

\begin{eqnarray}
\min \sum_j \lambda_j - \sum_{I}\kappa_I \nonumber\\
\lambda_j - \sum_{I\in (j,c)}\kappa_I\leq E_{j,c} & \hspace{2cm} \forall j,c \\
\lambda_j, \kappa_I\geq 0
\end{eqnarray}

We will show that the dual program can be solved in polynomial time by applying the Ellipsoid algorithm.
In order to do so, it suffices to construct a polynomial time separation oracle.
Assume that we are given a solution $(\lambda_j,\kappa_I)$ for the dual LP.
The separation oracle works as follows.
For each job $j\in\mathcal{J}$, we try to minimize the term $E_{j,c}+\sum_{I\in (j,c)}\kappa_I$.
If the value $\min_c\{E_{j,c}+\sum_{I\in (j,c)}\kappa_I\}$ is less than $\lambda_j$, then we have a violated constraint.
Otherwise, the solution is feasible.

In order to find the configuration that minimizes the above expression, we use dynamic programming.
Consider some configuration $c$.
The contribution of the operation $O_{j,k}$ in the expression $E_{j,c}+\sum_{I\in (j,c)}\kappa_I$
is the energy consumption of $O_{j,k}$ plus the $\kappa_I$'s of the intervals $I\in\mathcal{I}$ contained in the time slots during which $O_{j,k}$ is executed.
Let $A_{k,I}$ be the minimum contribution of the operations $O_{j,1},O_{j,2},\ldots,O_{j,k}$ to the objective function of our separation problem among the configurations
in which $O_{j,k}$ completes not later than $I$.
Furthermore, let $B_{k,I',I}$ be the minimum contribution of the operation $O_{j,k}$ to the objective function of separation problem
among the configurations in which it is executed after $I'$ and not later than $I$.
Clearly,
\begin{equation*}
A_{k,I}=\min_{I'<I}\{A_{k-1,I'}+B_{k,I',I}\}
\end{equation*}
In order to complete our dynamic programming algorithm, we have to specify a way of computing efficiently the term $B_{k,I',I}$.
Assume that there are $\ell$ time slots, between $I'$ and $I$, during which $O_{j,k}$ can be executed respecting Lemma \ref{lemma:discjs2}.
If we restrict our attention to configurations in which $O_{j,k}$ is executed in exactly $q \leq \ell$ slots,
$O_{j,k}$ must be executed during the $q$ slots with the minimum $\kappa_{I}$'s so that $E_{j,c}+\sum_{I\in (j,c)}\kappa_I$ is minimized.
These slots can be computed easily.
In order to compute $B_{k,I',I}$, it suffices to check all possible values $q=1,2,\ldots,\ell$.

Thus, for each job $j \in \mathcal{J}$,
we can compute, in polynomial time, the minimum $E_{j,c}+\sum_{I\in (j,c)}\kappa_I$ among all the configurations $c \in \mathcal{C}_j$,
and hence we have a polynomial-time separation oracle for the dual LP.
So, we can solve the dual LP in polynomial time with the Ellipsoid algorithm.
Given our discussion in Section \ref{section:migration}, we can solve the relaxed LP as well.
Then, by applying the same randomized rounding algorithm and analysis as in Section~\ref{section:rr}, we obtain the following theorem.

\begin{theorem}\label{thm:jobshop}
There is an algorithm of running time polynomial to $\mu$ and to $1/\varepsilon$
with approximation ratio $(1+\varepsilon)^{\alpha} (1+\frac{2}{\mu-2})^{\alpha} (1+\frac{\varepsilon}{1-\varepsilon})^{\alpha} \tilde{B}_{\alpha}$
for the preemptive job shop scheduling problem with the energy objective, where $\varepsilon\in(0,1)$.
\end{theorem}

%%%%%%%%%%%%%%%%%%%%%%%%%%%%%%%%%%%%%%%%%%%%%%%%%%%%%%%%%%%%%%%%%%%%%%%%%%%%%%%%%%%%%%%%%

\section{Routing}
\label{section:routing}

Now, we turn our attention to the min-power routing problem.
Formally, we are given a directed graph $G=(V,E)$ and a set of demands $\mathcal{D}$.
Each demand $i\in\mathcal{D}$ is associated with a source node $s_i$, a destination node $t_i$ and it requests $d_i$ integer units of bandwidth.
We consider the special case where all the demands request the same bandwidth, i.e. $d_{i}=d$ for all $i\in\mathcal{D}$.
Each edge $e \in E$ is associated with a constant $\alpha_e$ such that if $f$ units of demand cross $e$, then there is an energy consumption equal to $c_ef^{\alpha_e}$.
The objective is to route all the demands from their sources to their destinations so that the total energy consumption is minimized.
We consider the unsplittable version of the problem in which each demand has to be routed through a single path.

Andrews et al. \cite{AndrewsAZZ12} formulated the above problem as an integer convex program and they presented an analysis based on randomized rounding.
Using a similar analysis as in Section~\ref{section:rr}, we show that the algorithm presented in \cite{AndrewsAZZ12}
has a significantly better approximation ratio (see Table~\ref{tbl:results}).

In order to obtain the integer convex programming formulation in \cite{AndrewsAZZ12} for the min-power routing problem, we introduce a  variable $x_e$, for all $e\in E$, which corresponds to the number of demands that cross the edge $e$
and a binary variable $y_{i,e}$ which indicates if the demand $i\in\mathcal{D}$ crosses the edge $e$.
The integer convex program follows.

\begin{eqnarray}
\min \sum_{e\in E} c_ed^{\alpha_e}\max\{x_e,x_e^{\alpha_e}\} \nonumber\\
x_e=\sum_i y_{i,e} & \hspace{2cm} \forall e\in E \label{eqn:r1}\\
\sum_{e\in\Gamma^+(u)} y_{i,e} - \sum_{e\in\Gamma^-(u)} y_{i,e} = 0 & \hspace{2cm} \forall i \in \mathcal{D}, u\in V \setminus \{s_i,t_i\} \label{eqn:r2}\\
\sum_{e\in\Gamma^+(s_i)}y_{i,e}=1 & \hspace{2cm} \forall i\in\mathcal{D} \label{eqn:r3}\\
\sum_{e\in\Gamma^-(t_i)}y_{i,e}=1 & \hspace{2cm} \forall i\in\mathcal{D} \label{eqn:r4}\\
y_{i,e}\in\{0,1\} & \hspace{2cm} \forall i\in\mathcal{D},e\in E \label{eqn:r5}
\end{eqnarray}

The above integer convex program is a valid formulation for our problem.
Our goal is to minimize the total energy consumption of all edges, i.e. $\sum_{e\in E}c_ed^{\alpha_e} x_e^{\alpha_e}$.
Since all variables $x_e$ are integers in any feasible integral solution,
the above program has the same optimal integral solution as if we have used as objective the $\sum_{e\in E}c_ed^{\alpha_e} x_e^{\alpha_e}$.
However, the use of this objective leads to an integer program with large integrality gap \cite{AndrewsAZZ12}.
For this reason, we modify the objective to be $\sum_{e\in E}c_ed^{\alpha_e} \max\{x_e,x_e^{\alpha_e}\}$ obtaining a program with smaller integrality gap.
Equation~(\ref{eqn:r1}) relates the variables $x_e$ and $y_{i,e}$, while Equations~(\ref{eqn:r2})-(\ref{eqn:r4}) ensure the flow conservation.

In order to obtain a feasible integral solution for our problem, we solve the relaxation of the above convex program,
where the constraints $y_{i,e}\in\{0,1\}$ are relaxed so that $y_{i,e}\geq0$, and we obtain a fractional solution.
Then, we apply a randomized rounding procedure, introduced by Raghavan and Thompson \cite{RaghavanT91}, in order to select a path for each demand.
Specifically, for each demand $i\in\mathcal{D}$, we consider the subgraph of $G$ that contains only the edges with $y_{i,e}>0$
and define the standard flow decomposition.
We compute a $(s_i,t_i)$-path $p$ on this graph and we set $z_{i,p}=\min_{e\in p}\{y_{i,e}\}$.
Then, we subtract $z_{i,p}$ from the variables $y_{i,e}$ which correspond to the edges of the path $p$.
We continue this procedure until there are no $(s_i,t_i)$-paths.
The randomized rounding algorithm chooses a path $p$ for the demand $i$ with probability $z_{i,p}$.
Note that $\sum_{p}z_{i,p}=1$.

\begin{theorem}\label{thm:routing}
There is a $\tilde{B}_{\alpha_{max}}$-approximation algorithm for the min-power routing problem with uniform demands.
\end{theorem}
\begin{proof}
Consider an edge $e \in E$ and let $\lambda_e=\sum_{i \in \mathcal{D}} y_{i,e}$ be the expected number of demands that cross $e$.
The expected energy consumption on the edge $e$ is
\begin{equation}
\nonumber
E_e = c_ed^{\alpha_e}\sum_{S\subseteq \mathcal{D}}|S|^{\alpha_e} Pr(S)
\end{equation}
where $Pr(S)$ is the probability that exactly the demands in $S$ are routed through the edge $e$.
Hence, we have
\begin{eqnarray*}
E_e & = & c_ed^{\alpha_e}\sum_{S\subseteq \mathcal{D}}|S|^{\alpha_e} \prod_{i\in S} y_{i,e} \prod_{i \not\in S}(1-y_{i,e}).
\end{eqnarray*}
Since $y_{i,e}$ come from a mathematical programming solver, we can assume that there exists $N \in \mathbb{N}$ such that $y_{i,e}=\lambda_e\cdot \frac{q_{i,e}}{N}$ for some $q_{i,e} \in \mathbb{N}$.
Similarly with the proof of Theorem~\ref{thm:approx},   we can chop each $y_{i,e}$ into $q_{i,e}$ pieces $z_{i,e,\ell}=\frac{\lambda_e}{N}$.
Note that, $N = \sum_{i \in \mathcal{D}} q_{i,e}$ since $\sum_{i \in \mathcal{D}} \frac{y_{i,e}}{\lambda_e}=1$.
For the ease of exposition we identify the set $\{1,2,\ldots,N\}$ with the set of all pairs $((i,e),\ell)$ such that $i \in \mathcal{D}$ and $1 \leq \ell \leq q_{i,e}$.
By applying Proposition~\ref{prop:split} iteratively, we get
\begin{eqnarray*}
E_e & \leq & c_ed^{\alpha_e} \sum_{S\subseteq \{1,2,\ldots,N\}}|S|^{\alpha_e} \left(\frac{\lambda_e}{N}\right)^{|S|} \left(1-\frac{\lambda_e}{N}\right)^{N-|S|}\\
 & = & c_ed^{\alpha_e} \sum_{k=0}^N \sum_{S \in \{1,2,\ldots,N\}, |S|=k} k^{\alpha_e} \left(\frac{\lambda_e}{N}\right)^k \left(1-\frac{\lambda_e}{N}\right)^{N-k}\\
 & = & c_ed^{\alpha_e} \sum_{k=0}^N k^{\alpha_e} {N \choose k} \left(\frac{\lambda_e}{N}\right)^k \left(1-\frac{\lambda_e}{N}\right)^{N-k}
\end{eqnarray*}
as there are ${N \choose k}$ subsets of $\{1,2,\ldots,N\}$ with $k$ elements.
The sum in the last expression is the $\alpha_e$-th moment of a Binomial random variable $B_{\lambda_e}$ (sum of independent Bernoulli trials) with expectation $\lambda_e$.
Hence, by using Propositions~\ref{TechnicalPoisson} and~\ref{prop:bounds} we get
\begin{equation}
\nonumber
E_e \leq c_ed^{\alpha_e} \mathbb{E}[B_{\lambda_e}^{\alpha_e}]  \leq c_ed^{\alpha_e} \mathbb{E}[P_{\lambda_e}^{\alpha_e}] \leq c_ed^{\alpha_e}\max\{\lambda_e,\lambda_e^{\alpha_e}\} \mathbb{E}[P_{1}^{\alpha_e}] =  LP_e^* \tilde{B}_{\alpha_e}
\end{equation}
where $P_{\lambda_e}$ is a Poisson random variable with parameter $\lambda_e$.
By summing up over all edges and setting $\alpha=\max_{e\in E}\{\alpha_e\}$,
the theorem follows.
\end{proof}

\section{Conclusions}

We have presented a unified framework for dealing with various scheduling and routing problems in speed-scaling setting.
Our algorithms are based on configuration linear programs and randomized rounding.
Improving the approximation ratios or studying the inapproximability of the considered problems are some of the interesting directions for future work.
The most intriguing open question is the existence of a constant factor approximation algorithm for the non-preemptive multiprocessor scheduling problem.

\section*{Acknowledgements}

We would like to thank Oleg Pikhurko for providing the original proof of Part (b) of the Proposition \ref{prop:bounds}.


\begin{thebibliography}{10}

\bibitem{Albers10}
S.~Albers.
\newblock Energy-efficient algorithms.
\newblock {\em Communications of ACM}, 53:86--96, 2010.

\bibitem{Albers11}
S.~Albers.
\newblock Algorithms for dynamic speed scaling.
\newblock In {\em STACS}, pages 1--11, 2011.

\bibitem{AlbersAG11}
S.~Albers, A.~Antoniadis, and G.~Greiner.
\newblock On multi-processor speed scaling with migration: extended abstract.
\newblock In {\em SPAA}, pages 279--288. ACM, 2011.

\bibitem{AlbersMS07}
S.~Albers, F.~M\"{u}ller, and S.~Schmelzer.
\newblock Speed scaling on parallel processors.
\newblock In {\em SPAA}, pages 289--298. ACM, 2007.

\bibitem{AndrewsAZZ12}
M.~Andrews, A.~F. Anta, L.~Zhang, and W.~Zhao.
\newblock Routing for power minimization in the speed scaling model.
\newblock {\em IEEE/ACM Trans. on Networking}, 20:285--294, 2012.

\bibitem{AngelBKL12}
E.~Angel, E.~Bampis, F.~Kacem, and D.~Letsios.
\newblock Speed scaling on parallel processors with migration.
\newblock In {\em Euro-Par}, volume 7484 of {\em LNCS}, pages 128--140, 2012.

\bibitem{AntoniadisH12}
A.~Antoniadis and C.-C. Huang.
\newblock Non-preemptive speed scaling.
\newblock In {\em SWAT}, volume 7357 of {\em LNCS}, pages 249--260. Springer,
  2012.

\bibitem{AwerbuchKP92}
B.~Awerbuch, S.~Kutten, and D.~Peleg.
\newblock Competitive distributed job scheduling ({E}xtended abstract).
\newblock In {\em STOC}, pages 571--580, 1992.

\bibitem{BampisKLLN13}
E.~Bampis, A.~Kononov, D.~Letsios, G.~Lucarelli, and I.~Nemparis.
\newblock From preemptive to non-preemptive speed-scaling scheduling.
\newblock In {\em COCOON}, volume 7936 of {\em LNCS}, pages 134--146. Springer,
  2013.

\bibitem{BampisLL12}
E.~Bampis, D.~Letsios, and G.~Lucarelli.
\newblock Green scheduling, flows and matchings.
\newblock In {\em ISAAC}, volume 7676 of {\em LNCS}, pages 106--115. Springer,
  2012.

\bibitem{BaptisteCKQSS11}
P.~Baptiste, J.~Carlier, A.~Kononov, M.~Queyranne, S.~Sevastyanov,
and
  M.~Sviridenko.
\newblock Properties of optimal schedules in preemptive shop scheduling.
\newblock {\em Discrete Applied Mathematics}, 159(5):272--280, 2011.

\bibitem{BerendT10}
D.~Berend and T.~Tassa.
\newblock Improved bounds on {B}ell numbers and on moments of sums of random
  variables.
\newblock {\em Probability and Math. Statistics}, 30:185--205, 2010.

\bibitem{BinghamG08}
B.~D. Bingham and M.~R. Greenstreet.
\newblock Energy optimal scheduling on multiprocessors with migration.
\newblock In {\em ISPA}, pages 153--161. IEEE, 2008.

\bibitem{BrodtkorbDHHS10}
A.~R. Brodtkorb, C.~Dyken, T.~R. Hagen, J.~M. Hjelmervik, and
O.~O. Storaasli.
\newblock State-of-the-art in heterogeneous computing.
\newblock {\em Sci. Program.}, 18:1--33, 2010.

\bibitem{DengLX90}
X.~Deng, H.-N. Liu, and B.~Xiao.
\newblock Deterministic load balancing in computer networks.
\newblock In {\em SPDP}, pages 50--57, 1990.

\bibitem{GreinerNS09}
G.~Greiner, T.~Nonner, and A.~Souza.
\newblock The bell is ringing in speed-scaled multiprocessor scheduling.
\newblock In {\em SPAA}, pages 11--18. ACM, 2009.

\bibitem{GrotschelLS93}
M.~Gr\"{o}tschel, L.~Lov\'{a}sz, and A.~Schrijver.
\newblock {\em Geometric Algorithms and Combinatorial Optimizations, 2nd
  corrected edition}.
\newblock Springer-Verlag, 1993.

\bibitem{GuptaIKMP12}
A.~Gupta, S.~Im, R.~Krishnaswamy, B.~Moseley, and K.~Pruhs.
\newblock Scheduling heterogeneous processors isn't as easy as you think.
\newblock In {\em SODA}, pages 1242--1253, 2012.

\bibitem{Hoeffding56}
W.~Hoeffding.
\newblock On the distribution of the number of successes in independent trials.
\newblock {\em Annals of Mathematical Statistics}, 27:713--721, 1956.

\bibitem{RaghavanT91}
P.~Raghavan and C.~D. Thompson.
\newblock Randomized rounding: {A} technique for provably good algorithms and
  algorithmic proofs.
\newblock {\em Combinatorica}, 7:365--374, 1991.

\bibitem{Svensson11}
O.~Svensson.
\newblock Santa claus schedules jobs on unrelated machines.
\newblock In {\em STOC}, pages 617--626, 2011.

\bibitem{WiermanAT09}
A.~Wierman, L.~L.~H. Andrew, and A.~Tang.
\newblock Power-aware speed scaling in processor sharing systems.
\newblock In {\em INFOCOM}, pages 2007--2015, 2009.

\bibitem{YaoDS95}
F.~Yao, A.~Demers, and S.~Shenker.
\newblock A scheduling model for reduced {CPU} energy.
\newblock In {\em FOCS}, pages 374--382, 1995.

\end{thebibliography}
\end{document}